\documentclass[review]{ntupaper}
\usepackage[dvipsnames]{xcolor}
\usepackage{amssymb}
\usepackage{amsthm}
\newtheorem{remark}{Remark}
\newcommand{\Prob}{\mathbb{P}}

\newcommand{\up}{\textcolor{darkred}{$\uparrow$}}
\newcommand{\down}{\textcolor{BlueGreen}{$\downarrow$}}
\usepackage[absolute,overlay]{textpos}
\usepackage{tikz}
\usepackage{hyperref}
\usepackage{url}
\usepackage{threeparttable}
\usepackage{amsthm}
\usepackage[table]{xcolor}
\usepackage{hyperref}
\usepackage{booktabs}
\usepackage{multirow}
\usepackage[table]{xcolor}
\usepackage{fontawesome5}
\usepackage[smartEllipses]{markdown}
\usepackage{pifont}
\usepackage{xcolor}
\usepackage{bbm}
\usepackage[table]{xcolor} % For cell coloring
\usepackage{multirow}   % For multi-row cells
\usepackage{subcaption}
\usepackage{wrapfig}
\usepackage{xspace}
\usepackage{makecell}
\usepackage{wrapfig}
\usepackage{multirow}
\usepackage{graphicx,xcolor,float}
\usepackage{threeparttable}
\usepackage[ruled,vlined]{algorithm2e}
\usepackage{colortbl}
\usepackage{color}
\usepackage{multirow}
\usepackage{tabularx}
\usepackage{float}
\usepackage{tikz}
\usetikzlibrary{shapes.geometric, arrows.meta, positioning, calc, fit, backgrounds}
\usepackage{graphicx}
\usepackage{booktabs}
\usepackage{arydshln}
\usepackage{enumitem}
\usepackage{wrapfig}
\usepackage{caption}
\usepackage{graphicx}
\usepackage{makecell}
\usepackage{float} 
\usepackage{soul}
\usepackage{mathtools}
\usepackage{bbding}
\usepackage{makecell}
\usepackage{tabularx}
\newtheorem{theorem}{Theorem}
\newtheorem{corollary}{Corollary}
\newtheorem{assumption}{Assumption}
\newtheorem{proposition}{Proposition}
\newtheorem{lemma}{Lemma}
\newtheorem{definition}{Definition}
\usepackage{pifont}
\usepackage{xcolor}
\usepackage{bbm}
\usepackage[table]{xcolor} % For cell coloring
\usepackage{multirow}   % For multi-row cells

\usepackage{xspace}
\usepackage{seqsplit}

\usepackage[table,dvipsnames]{xcolor}
\usepackage{makecell}  % provides \Xhline

\usepackage{fontawesome5}

\usepackage{listings}

\tcbuselibrary{skins}

\definecolor{gtgray}{HTML}{F1F3F5}
\definecolor{posgreen}{HTML}{E6F4EA}
\definecolor{negred}{HTML}{FCE8E6}
\definecolor{bestgreen}{HTML}{C8E6C9}
\definecolor{darkred}{RGB}{128,0,32}

\definecolor{slate}{HTML}{536171}
\definecolor{arrowgray}{HTML}{52606D}
\definecolor{plateedge}{HTML}{A9B2BC}
\definecolor{backedge}{HTML}{C6CED6}
\definecolor{backfill}{HTML}{F5F7F8}
\definecolor{paperfill}{HTML}{DDE7EF}
\definecolor{gapfill}{HTML}{E3EFEA}
\definecolor{tealedge}{HTML}{6F9186}
\definecolor{questionfill}{HTML}{E8E4EF}
\definecolor{questionedge}{HTML}{7D718F}
\definecolor{thetafill}{HTML}{F2EBDD}
\definecolor{thetaedge}{HTML}{A79672}
\definecolor{panelfill}{HTML}{FAFAF7}
\definecolor{panelborder}{HTML}{D5D9DD}
\definecolor{notetext}{HTML}{39434D}

\usepackage{xcolor}
\usepackage{enumitem}

\definecolor{promptbg}{HTML}{F8FAFC}
\definecolor{promptframe}{HTML}{CBD5E1}

\newtcolorbox{promptbox}[1]{
  enhanced,
  breakable,
  colback=promptbg,
  colframe=promptframe,
  boxrule=0.45pt,
  arc=2mm,
  left=1.5mm,
  right=1.5mm,
  top=1mm,
  bottom=1mm,
  before skip=0.8em,
  after skip=0.8em,
  title=\textbf{#1},
  fonttitle=\small,
  coltitle=black
}

\usepackage{xcolor}
\usepackage{upgreek}

\usepackage{enumitem}

\usepackage{tikz}

\usetikzlibrary{positioning, arrows.meta}

\definecolor{rqblue}{HTML}{EAF2FF}
\definecolor{rqblueframe}{HTML}{2B5FAB}
\definecolor{gapred}{HTML}{FFF3F0}
\definecolor{gapredframe}{HTML}{B6472F}
\definecolor{ideagreen}{HTML}{EFFAF3}
\definecolor{ideagreenframe}{HTML}{2E7D4F}
\definecolor{softgray}{HTML}{F8FAFC}

\newtcolorbox{examplecard}[2]{
  enhanced,
  breakable,
  colback=#1,
  colframe=#2,
  boxrule=0.45pt,
  arc=2mm,
  left=1.5mm,
  right=1.5mm,
  top=1.2mm,
  bottom=1.2mm,
  before skip=0.5em,
  after skip=0.5em
}

\usepackage{xcolor}
\usepackage{enumitem}
\usepackage{tabularx}
\usepackage{ragged2e}
\usepackage{tikz}
\usetikzlibrary{arrows.meta}

\definecolor{gapbg}{HTML}{FFF8F5}
\definecolor{gapframe}{HTML}{C56A4A}
\definecolor{rqbg}{HTML}{F4F8FF}
\definecolor{rqframe}{HTML}{4A78B8}
\definecolor{ideabg}{HTML}{F5FBF6}
\definecolor{ideaframe}{HTML}{4A8F61}

\newtcolorbox{dapobox}[3]{
  enhanced,
  colback=#1,
  colframe=#2,
  boxrule=0.45pt,
  arc=2mm,
  left=2mm,
  right=2mm,
  top=1.5mm,
  bottom=1.5mm,
  height=#3,
  valign=top
}

\runningtitle{}

\runningstatus{Under review}
\title{$\Sigma$-Mem: An Online Reliability Memory for LLM-based Multi-Agent Systems}

\author[1]{Peilin Feng}
\author[1]{Suorong Yang}
\author[1]{Soujanya Poria}
\affiliation[1]{\small{DeCLaRe Lab, Nanyang Technological University}}
\github{Sigma-Mem}
\correspondence{Suorong Yang (\email{sryang@smail.nju.edu.cn}), Soujanya Poria (\email{soujanya.poria@ntu.edu.sg})}
\definecolor{forestgreen}{RGB}{34,139,34}

\abstract{
Memory is central to long-horizon LLM agents, yet existing memory systems
primarily preserve interaction content rather than modeling which agents can be
trusted and under what conditions. This limitation is particularly important in
multi-agent systems, where a central model may be unable to directly verify
plausible or correlated peer responses. We introduce \textbf{$\Sigma$-Mem}, an
online reliability memory that records \emph{historical competence evidence} for
individual peers and \emph{peer relationship evidence} across the peer set.
Both forms of evidence are maintained as real symmetric states and updated from
post-decision correctness feedback. By Weyl's inequality, the spectral change
caused by each event-level update is bounded, enabling stable online adaptation
without retraining the underlying models. $\Sigma$-Mem provides a general
write-and-read interface: the same memory can be used for residual steering of a
central model, response-free peer routing, or reliability-weighted voting. Across five Qwen-family models, $\Sigma$-Mem adapts to counterfactual
reliability shifts and generalizes to unseen peers and task domains.
Direct memory readouts also outperform majority voting and the best fixed peer
over the full OOD evaluation set. Moreover, performance improves consistently
as more correctness feedback becomes available, indicating that $\Sigma$-Mem
progressively accumulates actionable reliability information. These results
establish reliability memory as a reusable foundation for adaptive coordination
in LLM-based multi-agent systems.
}

\date{\today}

\usepackage{hyperref}

\usepackage{booktabs}
\usepackage{soul}
\usepackage[table]{xcolor}

\crefformat{section}{\S#2#1#3} % see manual of cleveref, section 8.2.1
\crefformat{subsection}{\S#2#1#3}
\crefformat{subsubsection}{\S#2#1#3}

\let\realcite\cite
\renewcommand{\cite}[1]{\ifx.#1.\hl{[?]}\else\realcite{#1}\fi}

\begin{document}

\maketitle

\section{Introduction}
% park2023generativeagents,chen2024reconcile
{Memory has become a pivotal component of LLM-based agents~\citep{zhang2025survey,wu2025human,hu2025memory}.
As tasks grow longer and contexts become bloated, agents rely on external memory to preserve and retrieve relevant information from historical interactions~\citep{packer2023memgpt,deltamem2026}.
Existing memory systems mainly focus on content memory: they store, retrieve, and summarize what was said, observed, or attempted in historical interactions.
This is useful in single-agent settings, where the main challenge is often preserving task-relevant information over long horizons.
However, LLM-based multi-agent systems (MAS) introduce a different memory problem~\citep{cemri2025mas}.
In MAS, a central model coordinates multiple peers~\citep{zhao2025language,khan2024debating, subramaniam2025multiagent, choi2026debate}. Many coordination decisions require judging not only what a peer said, but also whether that peer should be trusted for the current type of task.
Content memory alone does not handle this question: replaying or summarizing past dialogs does not explicitly reveal which peer was correct, on which tasks, under which conditions, or which peers tend to succeed or fail together. 
For example, a peer that is reliable on mathematical reasoning may still fail in retrieval question answering, while another peer that is strong at factual retrieval may fail in multi-step reasoning.
Besides, agreement among multiple peers may reflect shared biases or correlated errors rather than independent evidence.
This calls for reliability memory: a persistent, task-conditioned, peer-specific record of each agent's trustworthiness, continuously updated from past interactions and their outcomes. While content memory answers ``what happened,'' reliability memory answers ``who can be trusted, and when.''
}

{Reliability memory is fundamental to MAS because the central model can not always verify every peer response. Verification may be infeasible when supporting evidence is unavailable, and unreliable when judging an answer is as hard as producing it~\citep{kenton2024scalable, krumdick2025no}. Figure~\ref{fig:introduction} illustrates a typical case: in a RAG task where the supporting context is unavailable to the central model, multiple peers may provide confident and mutually consistent responses. 
Even a strong central model can remain uncertain when candidate answers are all plausible~\citep{simhi2025trust,lin2022truthfulqa} but difficult to verify~\citep{zheng2023judging}. 
% This uncertainty is further amplified when peers trained on similar data or optimized under similar objectives share biases and converge to the same incorrect consensus~\citep{chen2024reconcile}. 
In such cases, agreement among peers may indicate correlated errors rather than corroborating evidence~\citep{gradient2025mas,zhou2026ecl,sharma2024towards,zhou2026epistemic}.
Consequently, MAS failures often occur at trust decisions: for example, when a correct peer response is rejected while an incorrect one is accepted~\citep{cemri2025mas, pitre-etal-2025-consensagent, wynn2025talk}.
When direct verification is infeasible or expensive, the system's performance depends on whether it assigns attention to the right peers.
This is precisely the signal that reliability memory accumulates but content memory does not explicitly model.
}

\begin{figure}[]
\vspace{-0.5em}
\centering
\includegraphics[width=0.9\linewidth]{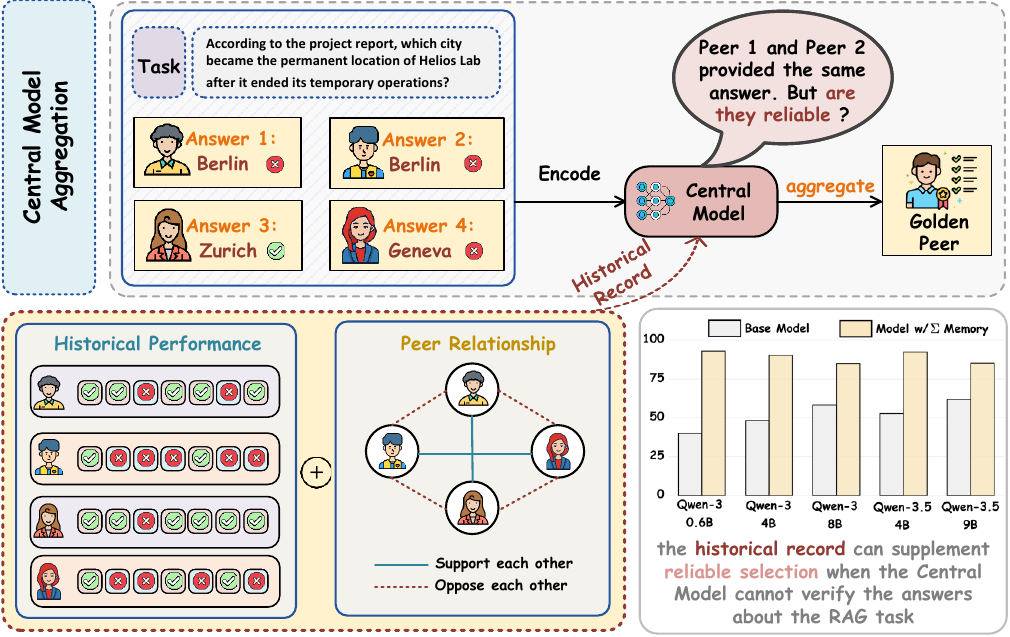}
\caption{Motivation of \textbf{$\Sigma$-Mem}. We hide the supporting context from the RAG task so that the central model cannot directly verify the correctness of the candidate answers. It is insufficient to aggregate peers for the central model when multiple peers provide confident or correlated answers. Historical performance and peer relationship can provide complementary reliability evidence to steer the central model toward better aggregation decisions.}
\label{fig:introduction}
\vspace{-1.5em}
\end{figure}

{In this paper, we propose \textbf{$\Sigma$-Mem}, an online memory mechanism that models and accumulates reliability evidence for LLM-based multi-agent systems. 
Instead of recording only interaction content, \textbf{$\Sigma$-Mem} maintains two complementary forms of reliability evidence.
The first is \textbf{\emph{historical competence evidence}}, which records how trustworthy each peer has been under different task conditions over time.
The second is \textbf{\emph{peer relationship evidence}}, which models how peers relate to each other through their correctness patterns, such as whether they tend to succeed independently or fail together due to shared biases or correlated reasoning errors.
These two forms of evidence play different but complementary roles: historical competence estimates whether an individual peer should be trusted, while peer relationship evidence helps determine whether agreement among peers reflects independent support or correlated failure.}

{Based on these reliability states, \textbf{$\Sigma$-Mem} defines a general write-and-read interface for trust-aware MAS memory. 
The write interface uses external correctness feedback to update the historical competence states $\mathbf{M}_p$ and the peer relationship state $\mathbf{G}$. 
Once written, the same memory can be read in different ways for different coordination decisions, such as steering response evaluation, routing tasks to suitable peers, or weighting peer outputs.
To make this memory adaptive over long horizons without allowing unstable feedback to distort the recorded reliability, we represent both $\{\mathbf{M}_p\}_{p=1}^{P}$ and $\mathbf{G}$ as real symmetric matrices and update them through a decayed historical state plus a bounded contribution from new responses~\citep{gu2021efficiently}. By Weyl's inequality $\big|\lambda(\mathbf{M}+\mathbf{E})-\lambda(\mathbf{M})\big|\leq\lVert\mathbf{E}\rVert_2$ ~\citep{weyl1912asymptotische, bhatia1997matrix}, the change of each eigenvalue after an update is bounded by the spectral norm of the new evidence, preventing abrupt distortions of the memory state.
This design allows repeated and consistent feedback to accumulate along stable spectral directions, while inconsistent or isolated signals are gradually attenuated.
Through this mechanism, \textbf{$\Sigma$-Mem} enables long-horizon MAS~\citep{zhang2026g,wu2026scaling,nayak2024llamar} to adapt their coordination behavior through online feedback without retraining the underlying models.
}

We evaluate \textbf{$\Sigma$-Mem} across five Qwen-family central models~\citep{yang2025qwen3, team2026qwen3}.
On the mixed counterfactual benchmark, \textbf{$\Sigma$-Mem} substantially improves peer selection when the reliable peer changes; especially at CF@90, it raises Qwen3-0.6B from 46.22\% to 71.10\%.
Although trained with only three peers, the system keeps improving when evaluated with four-peer and five-peer pools that contain unseen peer models.
Beyond the training domains of math, RAG, and code, it improves the base models in 27 out of 30 cases, with clear gains on BBH.
We then test whether the same memory supports decisions other than peer selection, as needed in agent swarms~\citep{zhuge2024language,zhang2025swarmagentic,team2026kimi} and mixture-of-agents architectures~\citep{wang2025mixture}.
Read as routing scores, the states pick a peer before any response is generated. This uses no central model at decision time, yet it outperforms both majority voting and the best fixed peer over the entire dataset. More importantly, We further conduct feedback ratio ablations on Qwen3.5-4B and Qwen3.5-9B. Accuracy on the OOD datasets consistently improves as the feedback ratio increases. This trend indicates that the \textbf{\emph{historical competence evidence}} maintained by $\Sigma$-Mem accumulates useful reliability information from observed events. These results further show that $\Sigma$-Mem is not tied to a specific aggregation procedure. Rather, it provides a reusable reliability memory that can support different selection mechanisms.

To sum up, our main contributions are summarized as follows:
\begin{itemize}

\item We propose \textbf{$\Sigma$-Mem}, an online reliability memory for MAS. It records \textbf{\emph{historical competence evidence}} for each peer and \textbf{\emph{peer relationship evidence}} across peers, both updated online from external correctness feedback.

\item We make the memory stable by design. Both states are real symmetric matrices with decayed, bounded updates. By Weyl's inequality, no single event can dominate the recorded reliability, while consistent evidence accumulates in the spectrum.

\item Across five central models, \textbf{$\Sigma$-Mem} adapts to counterfactual reliability shifts, generalizes to unseen peers and unseen domains. More importantly, we show that this memory serves many decision mechanisms. The same frozen memory state supports steered peer selection, response-free peer routing, and reliability-weighted voting well, without extra training.
\end{itemize}
\vspace{-0.5em}
\section{Background and Preliminaries}
\subsection{Weyl's Inequality}
A real symmetric matrix is a square matrix $\mathbf{M}\in\mathbb{R}^{r\times r}$ satisfying $\mathbf{M}=\mathbf{M}^\top$. Weyl's inequality characterizes the stability of eigenvalues under symmetric perturbations. Specifically, if a real symmetric matrix $\mathbf{M}$ is perturbed by another real symmetric matrix $\mathbf{E}$, then each eigenvalue satisfies
\begin{equation}
\label{eq:weyl}
\big|\lambda_i(\mathbf{M}+\mathbf{E})-\lambda_i(\mathbf{M})\big|
\le
\lVert \mathbf{E}\rVert_2 , \forall i=1,2,...,r
\end{equation}
where $\lVert \mathbf{E}\rVert_2$ is the spectral norm of the perturbation. The perturbation norm of a single event is usually bounded. This means that a small symmetric perturbation cannot cause an arbitrarily large change in any eigenvalue of the spectrum. The proof of Weyl's inequality is given in the appendix~\ref{appendix:weyl}.
In our \textbf{$\mathbf{\Sigma}$-Mem}, this property provides the basic spectral motivation for using the real symmetric memory matrix: each event-level update acts as a controlled perturbation, so the memory spectrum evolves stably while persistent aligned evidence can accumulate over time.
\subsection{Transformer Attention and Residual Steering}

In a Transformer~\citep{vaswani2017attention} for sequence modeling, let
$\mathbf{X}$ denote the hidden sequence. A self-attention block computes the query, key, and value matrices from the residual stream:
\begin{equation}
\label{eq:qkv}
\mathbf{Q}
=
\mathbf{W}_Q\operatorname{LN}(\mathbf{X}),\quad
\mathbf{K}
=
\mathbf{W}_K\operatorname{LN}(\mathbf{X}),\quad
\mathbf{V}
=
\mathbf{W}_V\operatorname{LN}(\mathbf{X}) .
\end{equation}
The attention output is then given by
\begin{equation}
\label{eq:attention}
\operatorname{Attn}(\mathbf{X})
=
\operatorname{softmax}
\left(
\frac{\mathbf{Q}\mathbf{K}^\top}{\sqrt{d_k}}
\right)
\mathbf{V}.
\end{equation}

Residual steering does not directly modify the attention parameters. Instead, it adds a peer-specific steer $\boldsymbol{\delta}_p(\mathbf{x})$ using the task input sequence $\mathbf{x}$ to the residual stream, thereby influencing the attention computation. Formally, we write
\begin{equation}
\label{eq:residual_perturb}
\widetilde{\mathbf{X}}_p
=
\mathbf{X}
+
\boldsymbol{\delta}_p(\mathbf{x}) .
\end{equation}
The attention block then computes its query, key, and value matrices from the new residual stream:
\begin{equation}
\label{eq:steered_qkv}
{\mathbf{Q}}_p
=
\mathbf{W}_Q
\operatorname{LN}\!\left(\widetilde{\mathbf{X}}_p\right),
\quad
{\mathbf{K}}_p
=
\mathbf{W}_K
\operatorname{LN}\!\left(\widetilde{\mathbf{X}}_p\right),
\quad
{\mathbf{V}}_p
=
\mathbf{W}_V
\operatorname{LN}\!\left(\widetilde{\mathbf{X}}_p\right).
\end{equation}
Thus, residual steering leaves the frozen attention weights unchanged, but changes the hidden states from which attention is computed. 
% In Section~\ref{sec:method}, we describe how $\Sigma$-Mem is designed to produce an effective peer-specific residual steering vector $\boldsymbol{\delta}_p(\mathbf{X})$.

\section{$\Sigma$-Mem Mechanism}
\label{sec:method}

\subsection{Event-Level $\Sigma$-Mem Record}
\label{sec:method:sigma_state}

\textbf{$\Sigma$-Mem} records two complementary forms of online evidence. The first is
\textbf{\emph{historical competence evidence}}, which tracks how each individual peer
performed on previous tasks close to the competence direction. The second is \textbf{\emph{peer relationship evidence}}, which tracks the relationships between peers across correctness. The former is recorded in peer-specific
symmetric matrices $\{\mathbf{M}_p\}_{p=1}^{P}$, while the latter is recorded in a symmetric peer relationship graph matrix $\mathbf{G}\in\mathbb{R}^{P\times P}$.

\begin{figure}[]
\vspace{-0.5em}
\centering
\includegraphics[width=1.0\linewidth]{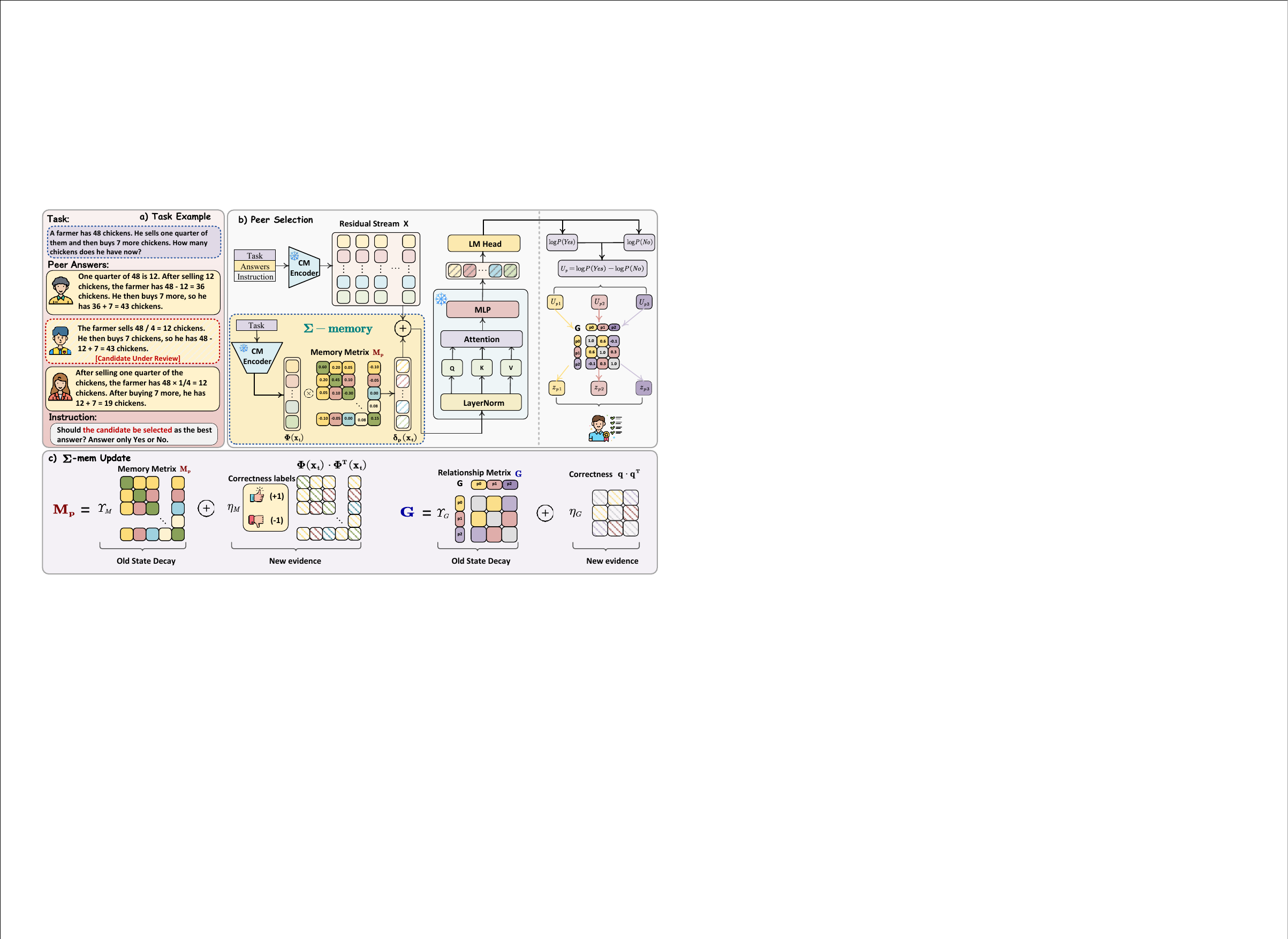}
\caption{Overview of \textbf{$\Sigma$-Mem}. Given a task and multiple peer answers: For each peer $p$, \textbf{$\Sigma$-Mem} reads the task-conditioned historical competence state $\mathbf{M}_p$ and converts it into a peer-specific residual steering signal, which influences the center model's attention and produces a utility score. The utilities are then combined with the peer relationship $\mathbf{G}$ for final aggregation. After external correctness labels become available, \textbf{$\Sigma$-Mem} updates both the competence memories and peer relationship matrix.}
\label{fig:method}
\end{figure}

As shown in Figure~\ref{fig:method}. For each peer $p$, \textbf{$\Sigma$-Mem} maintains a symmetric memory matrix
$\mathbf{M}_p\in\mathbb{R}^{r\times r}$, initialized as
$\mathbf{M}_p^{(0)}=\mathbf{0}$. \textbf{$\Sigma$-Mem} updates the state once per evaluated
event. Given an input $\mathbf{x}_t$, we compute a competence direction
$\boldsymbol{\phi}(\mathbf{x}_t)\in\mathbb{R}^r$ with
$\|\boldsymbol{\phi}(\mathbf{x}_t)\|_2=1$ and obtain a ground truth correctness label
$c_{p,t}\in\{+1,-1\}$ for peer $p$. The competence memory is updated as
\begin{equation}
\label{eq:sigma_mem_update}
\mathbf{M}_p^{(t+1)}
=
\gamma \mathbf{M}_p^{(t)}
+
\eta\, c_{p,t}\,
\boldsymbol{\phi}(\mathbf{x}_t)
\boldsymbol{\phi}(\mathbf{x}_t)^\top ,
\end{equation}
where $\gamma\in(0,1)$ is the decay factor and $\eta>0$ is the update factor. The update term
$\boldsymbol{\phi}(\mathbf{x}_t)\boldsymbol{\phi}(\mathbf{x}_t)^\top$
is a rank-one symmetric matrix. Therefore, starting from
$\mathbf{M}_p^{(0)}=\mathbf{0}$, the update preserves
$\mathbf{M}_p=\mathbf{M}_p^\top$ at every step. 
Intuitively, persistent aligned evidence accumulates in the spectrum of $\mathbf{M}_p$, while isolated or noisy evidence is bounded by Weyl's inequality (Eq.~\eqref{eq:weyl}) and gradually damped by the decay factor.
This bounded single step perturbation translates into a favorable long-horizon trade-off: task-aligned, persistent competence evidence accumulates in the spectrum  as $T$ grows.
\begin{theorem}[Persistent evidence dominates noise in the spectrum]
\label{thm:sigma-mem-snr-main}
Fix a competence direction $\boldsymbol{\phi}^\star$ along which peer $p$
has a persistent expected competence $\mu_p\in[-1,1]$, that is, $c_{p,t}=\mu_p+\xi_t$ for independent, zero-mean noise
$\xi_t\in[-2,2]$, whenever $\boldsymbol{\phi}(\mathbf{x}_t)=\boldsymbol{\phi}^\star$. Then for every
horizon $T$ and every $\delta\in(0,1)$, with probability at least $1-\delta$,
\begin{equation}
    \lambda_{\max}\bigl(\mathbf{M}_p^{(T)}\bigr)
    \;\geq\;
    \underbrace{\eta\mu_p\,\frac{1-\gamma^T}{1-\gamma}}_{\text{signal}}
    \;-\;
    \underbrace{\eta\sqrt{\dfrac{8\log(2/\delta)}{1-\gamma^2}}}
      _{\text{noise, independent of } T}.
\end{equation}
The noise term never grows with $T$, while the signal term saturates at
$\eta\mu_p/(1-\gamma)$; the resulting stationary signal-to-noise ratio
$\mu_p\sqrt{(1+\gamma)/(1-\gamma)}$ increases monotonically with the
decay factor $\gamma$.
\end{theorem}
The proof of Theorem~\ref{thm:sigma-mem-snr-main} is given in Appendix~\ref{appendix:proof_of_theorem1}, which shows that once $T$ is large enough for the signal term to clear the noise floor, the leading eigenvalue of $\mathbf{M}_p$ reflects peer $p$'s true competence rather than event-level noise, with high probability. 
This ensures that judgment is grounded in a genuine competence signal rather than historical noise.

In addition to the per-peer memory update, \textbf{$\Sigma$-Mem} records the relational structure among peers. Let
$\mathbf{c}_t=[c_{1,t},\ldots,c_{P,t}]^\top$ denote the correctness vector for event
$t$. We center this vector within the peer set:
\begin{equation}
\label{eq:sigma_centered_correctness}
q_{p,t}
=
c_{p,t}
-
\frac{1}{P}\sum_{j=1}^{P} c_{j,t}.
\end{equation}
The peer relationship matrix is then updated by
\begin{equation}
\label{eq:sigma_topology_update}
\mathbf{G}^{(t+1)}
=
\gamma_G \mathbf{G}^{(t)}
+
\eta_G \mathbf{q}_t\mathbf{q}_t^\top,
\qquad
\operatorname{diag}(\mathbf{G}^{(t+1)})=\mathbf{1},
\end{equation}
where $\gamma_G\in(0,1)$ and $\eta_G>0$ are the topology decay and update factors. Intuitively, $G_{p,q}>0$ indicates that the correctness signals of peers $p$ and $q$ tend to support each other within the peer set, whereas $G_{p,q}<0$ indicates that their correctness signals tend to oppose each other.

\subsection{Memory Readout and Residual Steering}
\label{sec:method:sigma_read}

At decision time for event $t$, \textbf{$\Sigma$-Mem} reads from the \textbf{\textit{historical competence evidence}}
$\mathbf{M}_p^{(t)}$. For each peer $p$, we compute a peer-specific memory readout
\begin{equation}
\label{eq:sigma_readout}
\mathbf{r}_{p,t}
=
\mathbf{M}_p^{(t)} \boldsymbol{\phi}(\mathbf{x}_t),
\end{equation}
where $\mathbf{r}_{p,t}\in\mathbb{R}^r$ is a competence reliability direction.
This readout is then projected to the hidden dimension of the center model:
\begin{equation}
\label{eq:sigma_delta}
\boldsymbol{\delta}_{p,t}(\mathbf{x}_t)
=
g\,\mathbf{P}\mathbf{r}_{p,t}
=
g\,\mathbf{P}\mathbf{M}_p^{(t)}\boldsymbol{\phi}(\mathbf{x}_t),
\end{equation}
where $\mathbf{P}\in\mathbb{R}^{d\times r}$ is a learned projection and $g$ is a
learned scalar gain. In order to preserve the center model's ability to understand the input while allowing
the memory to influence candidate evaluation, we inject the steering vector only
into the upper decoder blocks. Specifically, for
$\ell \in \{\lfloor L/2 \rfloor,\ldots,L-1\}$, the residual stream is shifted as
\begin{equation}
\label{eq:sigma_residual_shift}
\widetilde{\mathbf{X}}_{p}^{(\ell)}
=
\mathbf{X}^{(\ell)}
+
\boldsymbol{\delta}_{p,t}(\mathbf{x}_t).
\end{equation}

Following Eq.~\eqref{eq:steered_qkv}, each peer $p$ induces its own shifted residual
stream $\widetilde{\mathbf{X}}_{p}^{(\ell)}$. Because $\boldsymbol{\delta}_{p,t}(\mathbf{x}_t)$ depends on the peer-specific memory state $\mathbf{M}_p^{(t)}$, the same input
can lead to different residual shifts for different peers. The subsequent
attention computations are therefore peer-conditioned, enabling the center model
to distinguish candidates not only by their current answers, but also by their
historical reliability records.

\subsection{Utility Eval and Peer Selection}
\label{sec:method:sigma_scoring}

Given a set of peer answers $\{a_{p,t}\}_{p=1}^{P}$ for event $t$, the center
model evaluates each answer under its corresponding peer-specific steering vector
$\boldsymbol{\delta}_{p,t}(\mathbf{x}_t)$. The center model serves as a utility
evaluator: it is asked whether the answer from peer $p$ should be trusted for the
current task. We implement this utility as a Yes-versus-No likelihood contrast,
\begin{equation}
\label{eq:sigma_peer_utility}
U_{p,t}
=
\log P_{\mathrm{CM}}
\left(
\text{"Yes"}
\mid
\mathbf{x}_t, a_{p,t};
\boldsymbol{\delta}_{p,t}(\mathbf{x}_t)
\right)
-
\log P_{\mathrm{CM}}
\left(
\text{"No"}
\mid
\mathbf{x}_t, a_{p,t};
\boldsymbol{\delta}_{p,t}(\mathbf{x}_t)
\right).
\end{equation}
Here
$P_{\mathrm{CM}}(\cdot \mid \mathbf{x}_t,a_{p,t};
\boldsymbol{\delta}_{p,t})$
denotes the next-token distribution obtained by running the center model on the
judging prompt formed from $(\mathbf{x}_t,a_{p,t})$, while applying the
peer-specific residual shift $\boldsymbol{\delta}_{p,t}(\mathbf{x}_t)$ defined in
Sec.~\ref{sec:method:sigma_read}.

Since $U_{p,t}$ is computed under the residual steering induced by
$\mathbf{M}_p^{(t)}$, it already incorporates historical competence
evidence. To further incorporate the peer relationship, we normalize the
utilities within the peer set:
\begin{equation}
\label{eq:sigma_normalized_utility}
z_{p,t}
=
\frac{
U_{p,t}
-
\frac{1}{P}\sum_{q=1}^{P}U_{q,t}
}{
\sqrt{
\frac{1}{P}\sum_{q=1}^{P}
\left(
U_{q,t}
-
\frac{1}{P}\sum_{j=1}^{P}U_{j,t}
\right)^2
}
+\epsilon
}.
\end{equation}
We then define a posterior over latent peer correctness states
$\mathbf{y}_t\in\{-1,+1\}^{P}$:
\begin{equation}
\label{eq:sigma_topology_posterior}
P(\mathbf{y}_t \mid \mathbf{x}_t,\mathbf{M}^{(t)},\mathbf{G}^{(t)})
\propto
\exp\left(
w_u \sum_{p=1}^{P} z_{p,t} y_{p,t}
+
\frac{w_g}{2}
\sum_{p\neq q}
G_{p,q}^{(t)} y_{p,t}y_{q,t}
\right),
\end{equation}
where $w_u,w_g$ control
the strengths of the competence utility evidence and the peer relationship evidence,
respectively. The selected peer is the one with the largest posterior correctness expectation:
\begin{equation}
\label{eq:sigma_peer_select}
\hat{p}_t
=
\arg\max_p
\mathbb{E}
\left[
y_{p,t}
\mid
\mathbf{x}_t,\mathbf{M}^{(t)},\mathbf{G}^{(t)}
\right].
\end{equation}
These expectations can be
computed by exact enumeration over $\{-1,+1\}^{P}$. Importantly, $U_{p,t}$ is only a decision-time utility score and is not used as
the memory update label. After the event is evaluated, \textbf{$\Sigma$-Mem} receives the
external correctness signals $\{c_{p,t}\}_{p=1}^{P}$ from the benchmark or
environment and updates both $\mathbf{M}_p$ and $\mathbf{G}$ according to
Eq.~\eqref{eq:sigma_mem_update} and Eq.~\eqref{eq:sigma_topology_update}. This
separates selection from memory writing: peer selection is based on the center
model's current utility judgment, historical steering, and
peer relationship, while memory updates are grounded in external feedback. 

\subsection{Direct Readouts for Other Selection Mechanisms}
\label{sec:method:direct_readouts}

Broadly speaking, steered peer selection (Sec.~\ref{sec:method:sigma_scoring}) is one way to read the memory. The same states also support decisions that do not use the central model's judgment. At event $t$, the frozen central model encodes only the current question $\mathbf{x}_t$ to obtain the competence direction $\boldsymbol{\phi}(\mathbf{x}_t)$. We then compute the historical reliability score of peer $p$ as
\begin{equation}
\label{eq:direct_m_score}
s_{p,t}
=
\boldsymbol{\phi}(\mathbf{x}_t)^\top
\mathbf{M}_p^{(t)}
\boldsymbol{\phi}(\mathbf{x}_t).
\end{equation}
This score depends only on the current question and the task-conditioned competence evidence accumulated from previous events. And the peer's response is not used to compute $s_{p,t}$. This direct readout supports two decision rules.

\paragraph{$\mathbf{M}$ routing.}
The first rule selects the peer with the largest memory readout:
\begin{equation}
\label{eq:direct_m_route}
\widehat{p}_t^{\,\mathrm{M\text{-}route}}
=
\arg\max_p s_{p,t}.
\end{equation}
The final prediction is the response produced by the selected peer. The routing function itself does not receive the peer responses as input. Its decision is determined entirely by the historical competence evidence stored in $\{\mathbf{M}_p^{(t)}\}_{p=1}^{P}$.

\paragraph{$\mathbf{M}$-weighted voting.}
The second rule uses the direct memory readout as a peer-specific voting weight. For every canonical answer $a$, we aggregate the reliability scores of the peers that produce that answer:
\begin{equation}
\label{eq:direct_m_vote}
V_t^{\mathrm{M}}(a)
=
\sum_{p:\operatorname{canon}(a_{p,t})=a}
s_{p,t},
\qquad
\widehat{a}_t^{\,\mathrm{M\text{-}vote}}
=
\arg\max_a V_t^{\mathrm{M}}(a),
\end{equation}
where $\operatorname{canon}(\cdot)$ maps each response to its task-specific canonical answer. The reliability weights $\{s_{p,t}\}_{p=1}^{P}$ are computed before answer aggregation and do not depend on the content of the current responses. The response strings are used only to identify peers that provide the same canonical answer. They are not semantically encoded, verified, or judged by the central model.

\section{Experiments}
\subsection{Experimental Setup}
We train \textbf{\textbf{$\Sigma$-Mem}} on a mixed supervised dataset of 2,963 task events
constructed from three competence domains: \textbf{mathematical reasoning}, \textbf{retrieval-augmented question answering}, and \textbf{code generation}. Each event consists of a task input, a set of peer-generated answers, and external per-peer correctness labels. We evaluate \textbf{$\Sigma$-Mem} along four axes. First, a mixed counterfactual benchmark with four reliability splits (CF@0, CF@50, CF@70, and CF@90) tests whether the memory helps the center model adapt when the prior reliable peer changes; a higher CF level indicates a stronger counterfactual shift in peer reliability (Sec.~\ref{sec:CF_experiment}). Second, we enlarge the peer pool at test time with unseen peer models (Sec.~\ref{sec:Multi_Peers}). Third, we test generalization beyond the training domains on six OOD benchmarks covering commonsense reasoning, knowledge understanding, science question answering, complex reasoning, and language understanding (Sec.~\ref{sec:OOD_Task}). Finally, we test whether the same memory supports selection mechanisms beyond steered peer selection, including response-free M-route and reliability-weighted M-vote, together with a feedback-availability ablation (Sec.~\ref{sec:coordination-role}). The training and test sets are drawn from different data sources; their composition and construction are detailed in Appendix~\ref{appendix:dataset}.

For the center model, we use open-source models from the widely studied Qwen family, covering multiple model sizes: Qwen3-0.6B, Qwen3-4B, Qwen3-8B, Qwen3.5-4B, and Qwen3.5-9B~\citep{yang2025qwen3,team2026qwen3}. For peer answer generators, we use a heterogeneous set of models with complementary strengths: Gemma-3-4B-it~\citep{gemmateam2025gemma3technicalreport} for mathematical reasoning, Phi-4-mini-instruct~\citep{abouelenin2025phi} for retrieval-augmented question answering, and Qwen2.5-Coder-7B-Instruct~\citep{yang2025qwen3} for code generation. This peer pool creates clear competence differences across domains, making it suitable for evaluating whether \textbf{$\Sigma$-Mem} can help center model select and adapt among peers.

\subsection{Counter Factual Attack Experiment}
\label{sec:CF_experiment}
As shown in Table~\ref{tab:CF_results}, we compare Qwen3 and Qwen3.5 models with the same base models equipped with our proposed \textbf{$\Sigma$-Mem} on a mixed counterfactual benchmark with four reliability splits: CF@0, CF@50, CF@70, and CF@90. $\Sigma$ w/o G denotes \textbf{$\Sigma$-Mem} recording only \textbf{\emph{historical competence evidence}}, while $\Sigma$ w/ G denotes \textbf{$\Sigma$-Mem} incorporating \textbf{\emph{peer relationship evidence}}.

\begin{table*}[!h]
\centering
\caption{Counter Factual (CF) Attack results under different CF@ratios. \texttt{Acc} denotes the task accuracy obtained by using the selected peer response, and \texttt{p1(Gemma)/p2(Phi)/p3(Qwen-Coder)} reports the center model's selection counts for the three peer agents. We \textcolor{darkred}{\textbf{highlight}} the best performing result for each model.}
\label{tab:CF_results}
\setlength{\tabcolsep}{3pt}
\renewcommand{\arraystretch}{1.15}
\resizebox{\textwidth}{!}{
\begin{tabular}{lcccccccc}
\toprule
\multirow{2}{*}{Model}
& \multicolumn{2}{c}{CF@0}
& \multicolumn{2}{c}{CF@50}
& \multicolumn{2}{c}{CF@70}
& \multicolumn{2}{c}{CF@90} \\
\cmidrule(lr){2-3}
\cmidrule(lr){4-5}
\cmidrule(lr){6-7}
\cmidrule(lr){8-9}
& Acc & p1/p2/p3
& Acc & p1/p2/p3
& Acc & p1/p2/p3
& Acc & p1/p2/p3 \\
\midrule

% \rowcolor{gray!10}
% Oracle
% & 55.7\% & 2685/0/0
% & 62.9\% & 0/0/2685
% & 69.6\% & 0/0/2685
% & 75.8\% & 0/0/2685 \\

% \rowcolor{gray!10}
% Ceiling
% & 93.6\% & 1495/759/260
% & 93.6\% & 1200/859/455
% & 93.6\% & 1082/891/541
% & 93.6\% & 963/949/602 \\

% \midrule

Qwen3-0.6B
& 61.01\% & 1837/192/656
& 52.89\% & 1828/222/635
& 50.50\% & 1838/228/619
& 46.22\% & 1830/229/626 \\

\enspace + $\Sigma$ w/o G
& \textcolor{darkred}{\textbf{64.17\%}} & 1801/355/529
& 55.38\% & 617/511/1557
& 61.19\% & 261/714/1710
& 68.64\% & 82/968/1635 \\

\enspace + $\Sigma$ w/ G
& 63.87\% & 1817/347/521
& \textcolor{darkred}{\textbf{56.42}}\% & 614/396/1675
& \textcolor{darkred}{\textbf{62.72}}\% & 286/518/1881
& \textcolor{darkred}{\textbf{71.10}}\% & 81/611/1993 \\

\midrule

Qwen3-4B
& 61.71\% & 930/339/1416
& \textcolor{darkred}{\textbf{65.55}}\% & 930/429/1326
& 65.77\% & 944/460/1281
& 67.86\% & 955/482/1248 \\

\enspace + $\Sigma$ w/o G
& 63.35\% & 1434/589/662
& 61.15\% & 451/390/1844
& 66.41\% & 238/363/2084
& 72.74\% & 129/351/2205 \\

\enspace + $\Sigma$ w/ G
& \textcolor{darkred}{\textbf{64.51}}\% & 1290/542/853
& 63.61\% & 447/344/1894
& \textcolor{darkred}{\textbf{67.64}}\% & 265/300/2120
& \textcolor{darkred}{\textbf{73.30}}\% & 143/287/2255 \\

\midrule

Qwen3-8B
& 66.89\% & 1068/572/1045
& \textcolor{darkred}{\textbf{67.67}}\% & 1006/573/1106
& 66.26\% & 1008/566/1111
& 66.07\% & 985/545/1155 \\

\enspace + $\Sigma$ w/o G
& \textcolor{darkred}{\textbf{67.75}}\% & 1392/624/669
& 62.53\% & 535/547/1603
& 64.80\% & 307/552/1826
& 68.34\% & 227/571/1887 \\

\enspace + $\Sigma$ w/ G
& 67.67\% & 1384/558/743
& 63.76\% & 500/462/1723
& \textcolor{darkred}{\textbf{66.33}}\% & 230/493/1962
& \textcolor{darkred}{\textbf{69.16}}\% & 129/494/2062 \\

\midrule

Qwen3.5-4B
& 67.15\% & 1360/438/887
& \textcolor{darkred}{\textbf{66.15}}\% & 1271/526/888
& 64.69\% & 1258/563/864
& 63.87\% & 1211/565/909 \\

\enspace + $\Sigma$ w/o G
& 70.99\% & 1745/583/357
& 60.89\% & 463/385/1837
& 66.89\% & 179/457/2049
& 73.15\% & 70/696/1919 \\

\enspace + $\Sigma$ w/ G
& \textcolor{darkred}{\textbf{72.77}}\% & 1619/596/470
& 63.17\% & 451/331/1903
& \textcolor{darkred}{\textbf{68.60}}\% & 182/314/2189
& \textcolor{darkred}{\textbf{75.46}}\% & 71/288/2326 \\

\midrule

Qwen3.5-9B
& 66.48\% & 1041/582/1062
& \textcolor{darkred}{\textbf{68.94}}\% & 1144/528/1013
& 69.65\% & 1196/495/994
& 70.99\% & 1193/481/1011 \\

\enspace + $\Sigma$ w/o G
& 71.81\% & 1771/514/400
& 62.83\% & 469/359/1857
& 69.24\% & 210/343/2132
& 75.46\% & 137/329/2219 \\

\enspace + $\Sigma$ w/ G
& \textcolor{darkred}{\textbf{71.84}}\% & 1629/498/558
& 65.96\% & 459/316/1910
& \textcolor{darkred}{\textbf{71.40}}\% & 194/284/2207
& \textcolor{darkred}{\textbf{76.95}}\% & 95/214/2376 \\

\bottomrule
\end{tabular}
}
\end{table*}
\textbf{Obs.1.}
For Qwen3-0.6B, the base model shows a strong bias toward peer 1 across all CF ratios, selecting peer 1 for more than 1,800 examples in each CF setting. As the counterfactual attack becomes more severe, this fixed peer preference becomes increasingly harmful, causing the accuracy to drop from 61.01\% at CF@0 to 46.22\% at CF@90. After applying \textbf{$\Sigma$-Mem}, the model consistently improves performance under all CF ratios. In particular, $\Sigma$ w/ G increases accuracy from 46.22\% to 71.10\% at CF@90, showing that \textbf{$\Sigma$-Mem} \textbf{enables the weak center model to adjust which peer it should actually trust, thereby improving selection accuracy}.

\textbf{Obs.2.}
For Qwen3-4B, Qwen3-8B, Qwen3.5-4B, and Qwen3.5-9B, the base accuracy remains relatively stable as the counterfactual attack becomes more severe. After applying \textbf{$\Sigma$-Mem}, the models further improve at CF@0, CF@70, and CF@90 settings, showing that \textbf{$\Sigma$-Mem records peers' historical competence and helps the center model adjust which peer it should actually trust under clear reliability patterns}. The main exception is CF@50, where \textbf{$\Sigma$-Mem} decreases accuracy. This behavior is not a failure of the mechanism, but reflects the history-sensitive nature of \textbf{$\Sigma$-Mem} under an intentionally ambiguous data stream. In CF@50, half of the historical evidence supports the original peer preference, while the other half supports the counterfactual shift. As a result, the memory state accumulates misleading historical evidence from the CF@50 stream, which can steer the center model toward the wrong peer. We further verified this in Appendix~\ref{appendix:role}. We stress that this is a \emph{faithfulness} property of the memory rather than an aggregation-specific defect: \textbf{$\Sigma$-Mem} accurately integrates whatever reliability evidence the stream provides, so a stream that is ambiguous by construction (\textit{CF@50, where half the history supports each preference}) necessarily yields ambiguous memory. This motivates readouts that adaptively balance historical guidance against current-response evidence in future work.

\textbf{Obs.3.}
$\Sigma$ w/ G generally outperforms $\Sigma$ w/o G, indicating that the \textbf{\emph{peer relationship evidence}} provides useful complementary information beyond individual peer histories. This suggests that peer selection for the central model benefits not only from knowing whether each peer has been reliable on similar tasks, but also from modeling how peers are related to each other through their correctness patterns. Therefore, in the following experiments, including Experiment~2 in Sec.~\ref{sec:Multi_Peers} and Experiment~3 in
Sec.~\ref{sec:OOD_Task}, we utilize $\Sigma$ w/ G as the default \textbf{$\Sigma$-Mem}
configuration.

\subsection{Generalization to More Peers}
\label{sec:Multi_Peers}
Since the trainable components in \textbf{$\Sigma$-Mem} are shared across peers, adding new peers only requires allocating new memory states, without modifying the learned parameters. To test whether the learned selection mechanism can generalize beyond the peers utilized during training. During testing, we expand the peer set by introducing additional unseen peers, including \texttt{Peer 4} (Llama-3.2-3B-Instruct~\citep{grattafiori2024llama}) and \texttt{Peer 5} (BitCPM-CANN-3B~\citep{team2025minicpm4}), and evaluate the system under larger peer pools.
\begin{table}[!h]
\centering
\caption{Peer-count generalization accuracy results on the counter facual benchmark. \textbf{$\Sigma$-Mem} is trained with three peers and evaluated with enlarged peer pools containing 4 or 5 peers by adding unseen peer models at test time. Arrows indicate whether \textbf{$\Sigma$-Mem} improves(\up) or degrades(\down) the corresponding base model.}
\label{tab:scaling_results}
\setlength{\tabcolsep}{4pt}
\renewcommand{\arraystretch}{1.1}
\begin{tabular}{lcccccccc}
\toprule
\multirow{2}{*}{Center Model}
& \multicolumn{2}{c}{CF@0}
& \multicolumn{2}{c}{CF@50}
& \multicolumn{2}{c}{CF@70}
& \multicolumn{2}{c}{CF@90} \\
\cmidrule(lr){2-3}
\cmidrule(lr){4-5}
\cmidrule(lr){6-7}
\cmidrule(lr){8-9}
& Base & \textbf{$\Sigma$-Mem}
& Base & \textbf{$\Sigma$-Mem}
& Base & \textbf{$\Sigma$-Mem}
& Base & \textbf{$\Sigma$-Mem} \\
\midrule

\rowcolor{gray!10}
\multicolumn{9}{c}{\textbf{4 peers}} \\

Qwen3-0.6B
& {58.85\%} & {62.05\%}\up
& 49.83\% & 51.43\%\up
& 46.59\% & 56.69\%\up
& 43.24\% & 63.09\%\up \\

Qwen3-4B
& {60.48\%} & {62.94\%}\up
& 60.74\% & 57.17\%\down
& 60.60\% & 61.38\%\up
& 61.86\% & 67.49\%\up \\

Qwen3-8B
& {64.28\%} & {66.29\%}\up
& 61.94\% & 57.09\%\down
& 60.56\% & 60.11\%\down
& 60.11\% & 62.98\%\up \\

Qwen3.5-4B
& {63.50\%} & {70.54\%}\up
& 61.94\% & 56.91\%\down
& 60.56\% & 62.98\%\up
& 61.04\% & 67.97\%\up \\

Qwen3.5-9B
& {66.07\%} & {71.66\%}\up
& {65.33\%} & 59.59\%\down
& 65.18\% & 65.10\%\down
& 64.32\% & 71.21\%\up \\

\midrule

\rowcolor{gray!10}
\multicolumn{9}{c}{\textbf{5 peers}} \\

Qwen3-0.6B
& {62.79\%} & {63.72\%}\up
& 47.56\% & 51.88\%\up
& 40.93\% & 56.50\%\up
& 34.71\% & 62.09\%\up \\

Qwen3-4B
& {61.01\%} & {62.98\%}\up
& 59.66\% & 58.47\%\down
& 59.18\% & 61.68\%\up
& 59.18\% & 64.54\%\up \\

Qwen3-8B
& {64.28\%} & {66.55\%}\up
& 60.63\% & 56.72\%\down
& 58.99\% & 60.30\%\up
& 58.29\% & 60.97\%\up \\

Qwen3.5-4B
& {66.07\%} & {70.61\%}\up
& {63.65\%} & 58.21\%\down
& 62.53\% & 63.17\%\up
& 61.86\% & 68.19\%\up \\

Qwen3.5-9B
& {67.15\%} & {71.14\%}\up
& {65.25\%} & 60.60\%\down
& 64.80\% & 66.03\%\up
& 64.32\% & 69.09\%\up \\

\bottomrule
\end{tabular}
\end{table}

Consistent with Experiment~1(Sec~\ref{sec:CF_experiment}), \textbf{$\Sigma$-Mem} shows the same overall behavior under expanded peer pools. At CF@0, CF@70, and CF@90, the recorded \textbf{\emph{historical competence evidence}} and \textbf{\emph{peer relationship evidence}} \textbf{help the center model identify the appropriate peer, leading to improved accuracy over the base model}. The main exception remains CF@50. In this split, the data stream provides misleading historical evidence, which \textbf{$\Sigma$-Mem} faithfully accumulates and steers the center model toward the wrong peer.

\subsection{Generalization Beyond Training Domains}
\label{sec:OOD_Task}
Although \textbf{$\Sigma$-Mem} is trained only on math, RAG, and code events, it can map
an unseen task to the closest learned competence directions through $\boldsymbol{\phi}(\mathbf{x})$. To evaluate whether this ability supports generalization beyond the training domains, we test \textbf{$\Sigma$-Mem} on out-of-distribution(OOD) benchmarks covering physical commonsense reasoning (PIQA~\citep{bisk2020piqa}), broad knowledge understanding (MMLU~\citep{hendrycks2020measuring}), science question answering (OpenBookQA~\citep{mihaylov2018can} and SciQ~\citep{welbl2017crowdsourcing}), complex reasoning (BBH~\citep{suzgun2023challenging}), and language understanding (SuperGLUE~\citep{wang2019superglue}). The results are shown in Table~\ref{tab:generalization_results}.

\begin{table*}[h!]
\centering
\caption{Generalization results on out-of-distribution (OOD) domains. \textbf{$\Sigma$-Mem} is evaluated on unseen benchmarks covering commonsense reasoning, broad knowledge understanding, science question answering, complex reasoning, and language understanding. Arrows indicate whether \textbf{$\Sigma$-Mem} improves (\up) or degrades (\down) the corresponding base model.}
\label{tab:generalization_results}
\setlength{\tabcolsep}{3pt}
\renewcommand{\arraystretch}{1.1}
\resizebox{\textwidth}{!}{
\begin{tabular}{lcccccccccccc}
\toprule
\multirow{2}{*}{Center Model}
& \multicolumn{2}{c}{PIQA}
& \multicolumn{2}{c}{MMLU}
& \multicolumn{2}{c}{OpenBookQA}
& \multicolumn{2}{c}{SciQ}
& \multicolumn{2}{c}{BBH}
& \multicolumn{2}{c}{SuperGLUE} \\
\cmidrule(lr){2-3}
\cmidrule(lr){4-5}
\cmidrule(lr){6-7}
\cmidrule(lr){8-9}
\cmidrule(lr){10-11}
\cmidrule(lr){12-13}
& Base & \textbf{$\Sigma$-Mem}
& Base & \textbf{$\Sigma$-Mem}
& Base & \textbf{$\Sigma$-Mem}
& Base & \textbf{$\Sigma$-Mem}
& Base & \textbf{$\Sigma$-Mem}
& Base & \textbf{$\Sigma$-Mem} \\
\midrule

Qwen3-0.6B
& 80.69\% & 81.12\%\up
& 57.61\% & 61.20\%\up
& 74.00\% & 79.00\%\up
& 88.10\% & 88.50\%\up
& 25.25\% & 28.77\%\up
& 75.99\% & 80.38\%\up \\

Qwen3-4B
& 81.77\% & 82.92\%\up
& 60.29\% & 62.51\%\up
& 76.20\% & 78.40\%\up
& 89.90\% & 91.00\%\up
& 20.38\% & 28.66\%\up
& 78.18\% & 81.14\%\up \\

Qwen3-8B
& 82.92\% & 83.57\%\up
& 60.42\% & 62.90\%\up
& 79.60\% & 79.00\%\down
& 91.30\% & 91.70\%\up
& 19.29\% & 28.17\%\up
& 79.00\% & 81.07\%\up \\

Qwen3.5-4B
& 83.19\% & 83.46\%\up
& 61.20\% & 62.90\%\up
& 79.60\% & 81.40\%\up
& 92.00\% & 92.80\%\up
& 27.46\% & 29.15\%\up
& 79.59\% & 80.76\%\up \\

Qwen3.5-9B
& 84.55\% & 83.46\%\down
& 62.83\% & 64.27\%\up
& 79.80\% & 80.20\%\up
& 92.20\% & 92.10\%\down
& 26.80\% & 29.33\%\up
& 80.79\% & 81.67\%\up \\

\bottomrule
\end{tabular}
}
\end{table*}

\textbf{$\Sigma$-Mem improves over the corresponding Qwen base model on most (27/30) OOD cases}. The improvement is especially pronounced on BBH, where all center models benefit substantially from \textbf{$\Sigma$-Mem}. For example, Qwen3-4B improves from 20.38\% to 28.66\%, and Qwen3-8B improves from 19.29\% to 28.17\%. These results indicate that the learned competence directions and peer-reliability evidence provide transferable guidance for peer selection on unseen task distributions.

\subsection{Generalization Across MAS Selection Mechanisms}
\label{sec:coordination-role}
In this section, we show that the proposed reliability memory generalizes across selection mechanisms, the different ways the MAS consumes reliability evidence, such as routing a task to a peer or weighted voting among peers.

We evaluate the direct readouts of Sec.~\ref{sec:method:direct_readouts} on the same OOD datasets and three fixed peers described in Sec.~\ref{sec:OOD_Task}, using all five Qwen central models. As a memory-free control, we use majority voting, which assigns the same weight to every peer and selects the answer supported by the largest number of peers:
\begin{equation}
\label{eq:direct_majority}
\widehat{a}_t^{\,\mathrm{Maj}}
=
\arg\max_a
\sum_{p:\operatorname{canon}(a_{p,t})=a}
\frac{1}{P}.
\end{equation}
Majority voting uses neither $\mathbf{M}_p$ nor a central model utility score. It only determines whether different peers produce the same canonical answer. For all methods, the correctness labels of the current event are accessed only after the predictions have been made. These labels are then used to update the memory for subsequent events. This comparison isolates whether the historical competence evidence stored in $\mathbf{M}_p$ provides reliable evidence beyond unweighted peer agreement. Table~\ref{tab:coordination_roles} reports overall accuracy; per-benchmark results are in Appendix~\ref{appendix:validation}.

\begin{table}[!h]
\centering
\caption{Overall accuracy of different selection mechanisms across five central models
and six OOD benchmarks. \texttt{M-Route} directly selects the peer with the
highest memory-derived reliability score, whereas \texttt{M-Vote} uses these
scores to weight peer votes during answer aggregation. Both methods rely only
on memory readouts and do not require the central model to evaluate peer
responses at decision time. The lower block reports model-independent routing
and voting baselines.}
\label{tab:coordination_roles}
\setlength{\tabcolsep}{4pt}
\renewcommand{\arraystretch}{1.1}

\resizebox{\linewidth}{!}{
\begin{tabular}{
@{}l
*{4}{>{\centering\arraybackslash}p{1.55cm}}
@{}
}
\toprule
\multirow{2}{*}{Central Model}
& \multicolumn{2}{c}{Selection (responses)}
& Routing
& Voting \\
\cmidrule(lr){2-3}
\cmidrule(lr){4-4}
\cmidrule(lr){5-5}
& Base
& $\Sigma$ w/ G
& M Route
& M Vote \\
\midrule

Qwen3-0.6B
& 56.52\%
& 59.89\%
& 60.67\%
& \textbf{60.93\%} \\

Qwen3-4B
& 55.98\%
& 60.54\%
& 60.68\%
& \textbf{60.96\%} \\

Qwen3-8B
& 56.16\%
& 60.50\%
& 60.67\%
& \textbf{60.99\%} \\

Qwen3.5-4B
& 59.56\%
& \textbf{60.87\%}
& 60.67\%
& 60.86\% \\

Qwen3.5-9B
& 60.04\%
& \textbf{61.31\%}
& 60.69\%
& 60.86\% \\

\midrule
Random peer (expected; uniform selection)
& & & 54.85\% & -- \\

Majority voting~\citep{wang2022self}
& & & -- & 59.12\% \\

% Beta reputation (decayed success rate; \citealp{josang2002beta})
% & & & 60.68\% & -- \\

Best fixed peer (overall best peer; \citealp{cesa2006prediction})
& & & 57.52\% & -- \\

Oracle reputation (per-subtest best peer; cf.~\citealp{dawid1979maximum})
& & & 60.79\% & -- \\

\bottomrule
\end{tabular}
}
\end{table}

\begin{figure}[!h]
\vspace{-1em}
\centering
\includegraphics[width=0.9\linewidth]{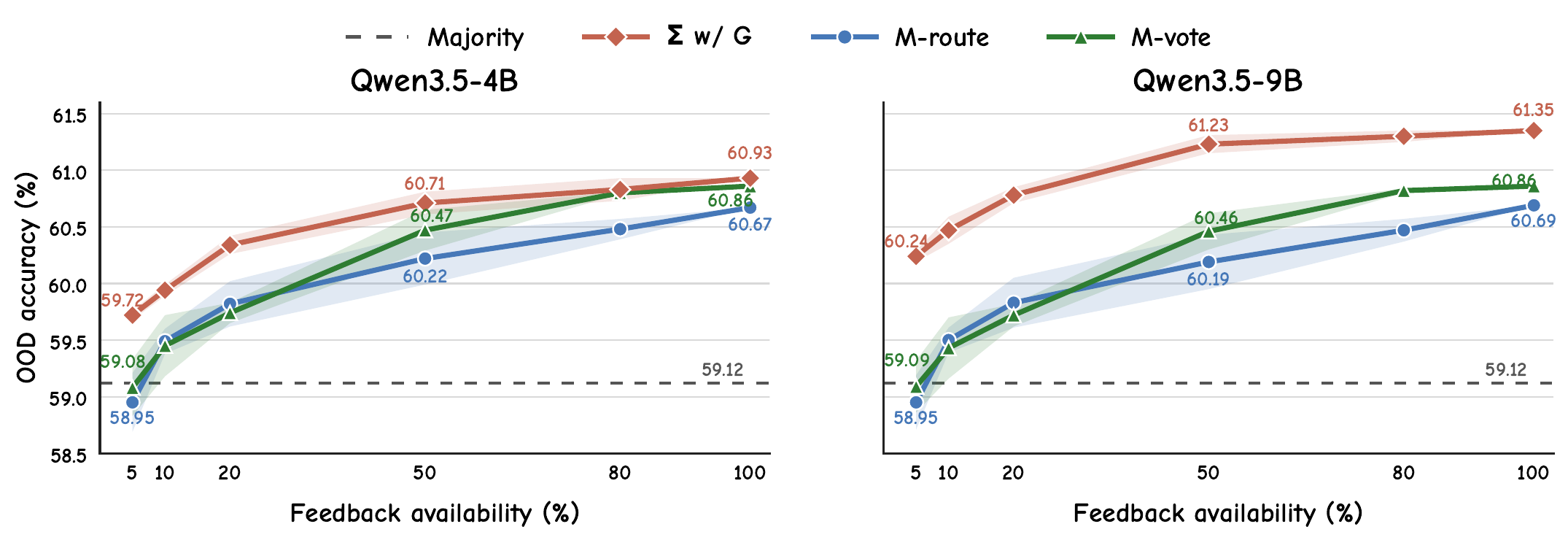}
\caption{OOD accuracy under different feedback availability ratios, with Qwen3.5-4B and Qwen3.5-9B as the central models. We vary the fraction of events whose correctness feedback is revealed after prediction, from 5\% to 100\%. Unlabeled events still apply temporal decay, which makes low-feedback settings strictly harder. The \textcolor{blue}{blue curve} denotes $\mathbf{M}$-based routing, the \textcolor{darkred}{red curve} denotes $\Sigma$-Mem w/ $\mathbf{G}$, and the \textcolor{gray}{gray dashed line} denotes majority voting. Shaded regions indicate variation across three random feedback masks.}
\vspace{-0.8em}
\label{fig:feedback_availability}
\end{figure}

As shown in Table~\ref{tab:coordination_roles}, directly using the
\textbf{\emph{historical competence evidence}} for simple routing and weighted
voting substantially improves OOD performance, outperforming both majority
voting and the best fixed peer baselines over the entire dataset. This demonstrates that
\textbf{$\Sigma$-Mem is not tied to a particular aggregation procedure; rather, it
provides a reusable reliability memory that can support different selection mechanisms}. The comparison also reveals the complementary roles of memory and
the central model. While \texttt{M-Vote} relies only on historical reliability,
residual steering combines this reliability information with the central model's assessment of the current responses. For the Qwen3 models, \texttt{M-Vote} performs better, suggesting that direct memory readout is more
effective when response-level assessment is limited. In contrast, with the
stronger Qwen3.5 models, $\Sigma$ w/G achieves the best overall accuracy, indicating that a more capable central model can better integrate current response evidence with historical reliability guidance.
The counterfactual streams complete this picture (Appendix~\ref{appendix:role}). When the reliability pattern is stable, as in CF@0 and CF@90, M-route matches or exceeds $\Sigma$ w/G for most central models. When history is ambiguous, as in CF@50, M-route degrades the most and the central model's response-level judgment becomes essential.

We further evaluate how memory performance changes as the availability of correctness feedback increases (Figure~\ref{fig:feedback_availability}). During the OOD test stream, each method first makes its prediction using the current memory state. After prediction, ground-truth correctness labels are revealed only for a randomly selected subset of events, with feedback availability varied from 5\% to 100\%. Events without feedback do not write new correctness evidence into memory and only apply temporal decay to the existing state. For each feedback ratio, we repeat the experiment with three random seeds. The results show that both memory-based routing and $\Sigma$-Mem improve as more feedback becomes available, indicating that the \textbf{memory state $\mathbf{M}$ records useful signals from the OOD stream that help peer selection}. Moreover, for Qwen 3.5, $\Sigma$-Mem consistently outperforms $\mathbf{M}$-route, suggesting that the central model’s content-based understanding provides complementary decision evidence beyond direct memory readout. This indicates that \textbf{by leveraging this inherent judgment capability, $\Sigma$-Mem further
improves overall accuracy.}
\section{Related Work}
\textbf{Computational Trust and Reputation in MAS.}
Maintaining per-agent trust from interaction outcomes has a long history in MAS. TRAVOS~\citep{teacy2006travos} estimates probabilistic trust from interaction track records; and
FIRE~\citep{Huynh2004FIREAI} integrates multiple trust sources, including
role-based and certified reputation; and REGRET~\citep{sabater2001regret}
incorporates social and contextual dimensions of reputation. 
This line established trust as a first-class signal for agent coordination, primarily with statistical or rule-based representations over predefined interaction categories.

% \textbf{Competence-aware Answer Aggregation.}
% The answers can be weighted by estimated annotator competence~\citep{dawid1979maximum} or conditioned by annotator ability on item type~\citep{whitehill2009whose}.
% In online learning with expert advice, Weighted Majority~\citep{littlestone1994weighted} maintains multiplicative reliability weights with regret guarantees, and the specialists framework~\citep{freund1997using} extends this to experts that are competent only on subsets of instances, which is an early form of task-conditioned reliability.

\textbf{LLM Routing and Ensembling.}
In the LLM era, competence estimates are used to select or combine models. RouteLLM~\citep{hallgarten2025routellm} and FrugalGPT~\citep{chen2023frugalgpt} train routers that dispatch queries to models by predicted quality and cost.
Zooter~\citep{lu2024routing} distills reward signals into a query-level router. Moreover, LLM-Blender~\citep{jiang2023llm} ranks and fuses candidate outputs.
For multi-agent settings, ReConcile~\citep{chen2024reconcile} performs confidence-weighted voting among heterogeneous LLM peers across discussion rounds, and GPTSwarm~\citep{zhuge2024gptswarm} formulates agent collectives as graphs and optimizes inter-agent edge weights from utility feedback.
Beyond mechanism design, recent studies examine trust dynamics in LLM-based MAS directly: LLM agents are prone to conformity, abandoning correct answers under peer pressure in social interactions~\citep{song2025llmscanthandlepeer}, and epistemic context learning~\citep{zhou2026ecl} studies how agents should calibrate trust in peers from interaction context. 

\textbf{Memory for LLM Agents.}
A parallel line equips agents with persistent memory.
MemGPT~\citep{packer2023memgpt} manages memory as a virtual context hierarchy; MemoryBank~\citep{zhong2024memorybank} and Mem0~\citep{chhikara2025mem0} maintain long-term records across sessions. Meanwhile, DeltaMem~\citep{deltamem2026} and proactive memory agents~\citep{wu2026remember} decide what to retain and when remembered state should re-enter the context during long-horizon execution. These systems concern the persistence and reactivation of \emph{content}.

\section{Conclusion}
We introduced \textbf{$\Sigma$-Mem}, an online reliability memory that records \textbf{\emph{historical competence evidence}} for individual peers and \textbf{\emph{peer relationship evidence}} across the peer set. Its symmetric state design supports stable online updates, while its general write-and-read interface allows the same reliability memory to guide residual steering, peer routing, and weighted voting without retraining the underlying models. Taken
together, our results show that reliability memory is not merely an auxiliary signal for response aggregation but a reusable coordination state that transfers across peer sets, task distributions, and selection mechanisms. They also reveal an important boundary: when the observed reliability evidence is ambiguous, as in CF@50, persistent memory provides limited guidance. Developing adaptive readouts that determine how historical reliability should interact with current decision evidence is therefore a promising direction. Overall, $\Sigma$-Mem establishes reliability memory as a foundation for adaptive coordination in long-horizon multi-agent systems.

% This is the theory part for bounds.
\bibliographystyle{plainnat}
\bibliography{citations}
\clearpage
\appendix

\section{Theoretical Analysis}\label{appendix:theoretical}
\subsection{Proof of Weyl's Inequality}
\label{appendix:weyl}
We provide a proof of the eigenvalue perturbation bound used in
Eq.~\ref{eq:weyl}. Let $\mathbf{M},\mathbf{E}\in\mathbb{R}^{r\times r}$
be real symmetric matrices, and order the eigenvalues of any real symmetric
matrix $\mathbf{A}$ as
\[
\lambda_1(\mathbf{A})
\ge
\lambda_2(\mathbf{A})
\ge
\cdots
\ge
\lambda_r(\mathbf{A}).
\]

According to the Courant--Fischer variational principle
~\citep{bhatia1997matrix}, the $i$-th eigenvalue of $\mathbf{A}$ satisfies
\begin{equation}
\label{eq:courant_fischer}
\lambda_i(\mathbf{A})
=
\min_{\substack{
\mathcal{S}\subseteq\mathbb{R}^{r}\\
\dim(\mathcal{S})=r-i+1
}}
\;
\max_{\substack{
\mathbf{v}\in\mathcal{S}\\
\lVert\mathbf{v}\rVert_2=1
}}
\mathbf{v}^{\top}\mathbf{A}\mathbf{v}.
\end{equation}

For every unit vector $\mathbf{v}$, the definition of the spectral norm gives
\begin{equation}
\label{eq:rayleigh_perturbation_bound}
\left|
\mathbf{v}^{\top}\mathbf{E}\mathbf{v}
\right|
\le
\lVert\mathbf{v}\rVert_2
\lVert\mathbf{E}\mathbf{v}\rVert_2
\le
\lVert\mathbf{E}\rVert_2.
\end{equation}
Therefore,
\begin{align}
\lambda_i(\mathbf{M}+\mathbf{E})
&=
\min_{\substack{
\mathcal{S}\subseteq\mathbb{R}^{r}\\
\dim(\mathcal{S})=r-i+1
}}
\;
\max_{\substack{
\mathbf{v}\in\mathcal{S}\\
\lVert\mathbf{v}\rVert_2=1
}}
\mathbf{v}^{\top}
(\mathbf{M}+\mathbf{E})
\mathbf{v}
\nonumber\\
&\le
\min_{\substack{
\mathcal{S}\subseteq\mathbb{R}^{r}\\
\dim(\mathcal{S})=r-i+1
}}
\;
\max_{\substack{
\mathbf{v}\in\mathcal{S}\\
\lVert\mathbf{v}\rVert_2=1
}}
\left(
\mathbf{v}^{\top}\mathbf{M}\mathbf{v}
+
\lVert\mathbf{E}\rVert_2
\right)
\nonumber\\
&=
\lambda_i(\mathbf{M})
+
\lVert\mathbf{E}\rVert_2.
\label{eq:weyl_upper_bound}
\end{align}

To obtain the reverse bound, we write
$\mathbf{M}=(\mathbf{M}+\mathbf{E})+(-\mathbf{E})$ and apply
Eq.~\ref{eq:weyl_upper_bound} again:
\begin{align}
\lambda_i(\mathbf{M})
&\le
\lambda_i(\mathbf{M}+\mathbf{E})
+
\lVert-\mathbf{E}\rVert_2
\nonumber\\
&=
\lambda_i(\mathbf{M}+\mathbf{E})
+
\lVert\mathbf{E}\rVert_2.
\end{align}
Hence,
\begin{equation}
\lambda_i(\mathbf{M})
-
\lVert\mathbf{E}\rVert_2
\le
\lambda_i(\mathbf{M}+\mathbf{E})
\le
\lambda_i(\mathbf{M})
+
\lVert\mathbf{E}\rVert_2.
\end{equation}
Combining the two inequalities yields
\begin{equation}
\label{eq:weyl_proof_result}
\left|
\lambda_i(\mathbf{M}+\mathbf{E})
-
\lambda_i(\mathbf{M})
\right|
\le
\lVert\mathbf{E}\rVert_2,
\qquad
i=1,\ldots,r,
\end{equation}
which proves Eq.~\eqref{eq:weyl}.

\subsection{Proof of Theorem~\ref{thm:sigma-mem-snr-main}}
\label{appendix:proof_of_theorem1}
This subsection proves Theorem~\ref{thm:sigma-mem-snr-main} from the main text. The proof proceeds in three steps: (1) a single-step spectral
perturbation bound, (2) a general signal-noise decomposition of the projected memory under an arbitrary sequence of task directions, and (3) specializing this decomposition to persistent alignment with a fixed competence direction and combining it with the variational characterization of eigenvalues to conclude the proof.

\textbf{Step 1: Single-Step Spectral Perturbation Bound}
\begin{proposition}[Bounded perturbation per event]
\label{prop:single-step-weyl}
For every peer $p$ and step $t$, every eigenvalue of $\mathbf{M}_p^{(t+1)}$
satisfies
\begin{equation}
    \bigl|\lambda_i\bigl(\mathbf{M}_p^{(t+1)}\bigr)
    - \gamma\,\lambda_i\bigl(\mathbf{M}_p^{(t)}\bigr)\bigr|
    \;\leq\; \eta.
\end{equation}
\end{proposition}

\begin{proof}
Write $\mathbf{M}_p^{(t+1)} = \gamma\mathbf{M}_p^{(t)} + \mathbf{E}_t$ with $\mathbf{E}_t = \eta\,c_{p,t}\,\boldsymbol{\phi}(\mathbf{x}_t)
\boldsymbol{\phi}(\mathbf{x}_t)^\top$. Since $\|\boldsymbol{\phi}(\mathbf{x}_t)\|_2=1$ and $\mathbf{E}_t$ is rank one, its only nonzero eigenvalue is $\eta\,c_{p,t}\in\{-\eta,\eta\}$, so $\|\mathbf{E}_t\|_2 = \eta$. Both $\gamma\mathbf{M}_p^{(t)}$ and $\mathbf{E}_t$ are symmetric, so Weyl's inequality (Eq.~\ref{eq:weyl}) applies directly with $\mathbf{M} = \gamma\mathbf{M}_p^{(t)}$ and $\mathbf{E} = \mathbf{E}_t$.
\end{proof}
This confirms that no single event, however extreme, can move any eigenvalue of the (decayed) memory by more than $\eta$. The remaining steps show what this buys us over a long horizon.

\textbf{Step 2: A General Signal-Noise Decomposition}
\begin{assumption}[Task-conditional competence model]
\label{assumption:competence-model}
Fix a unit reference direction $\boldsymbol{\phi}^\star\in\mathbb{R}^r$ (e.g., a task type of interest) and write $\psi_t = \bigl(\boldsymbol{\phi}(\mathbf{x}_t)^\top \boldsymbol{\phi}^\star\bigr)^2 \in [0,1]$ for the squared alignment of the event-$t$ task direction with $\boldsymbol{\phi}^\star$. The correctness label decomposes as
\begin{equation}
    c_{p,t} \;=\; \mu_p(\mathbf{x}_t) \;+\; \xi_t,
\end{equation}
where $\mu_p(\mathbf{x}_t) \in [-1,1]$ is peer $p$'s (unknown, task-dependent) expected competence and $\xi_{0}, \xi_1, \dots$ are independent, zero-mean noise terms with $\xi_t \in [-2,2]$ almost surely.
\end{assumption}

This is the minimal model needed to separate ``true, task-dependent competence'' from ``event-level noise'' while keeping $c_{p,t}\in\{-1,1\}$ exactly, since $\xi_t = c_{p,t} - \mu_p(\mathbf{x}_t)$ automatically lies in $[-2,2]$.

\begin{definition}[Projected memory signal]
For a fixed unit direction $\boldsymbol{\phi}^\star$, define the scalar projection
\begin{equation}
    s_T \;=\; \boldsymbol{\phi}^{\star\top}\mathbf{M}_p^{(T)}
    \boldsymbol{\phi}^\star.
\end{equation}
\end{definition}

\begin{lemma}[Signal accumulates, noise stays bounded]
\label{lem:signal-noise}
Under Assumption~\ref{assumption:competence-model}, with $\mathbf{M}_p^{(0)}=\mathbf{0}$, the projected memory signal decomposes as
\begin{equation}
    s_T \;=\; \underbrace{\eta\sum_{t=0}^{T-1}\gamma^{T-1-t}\,
    \psi_t\,\mu_p(\mathbf{x}_t)}_{\textstyle S_T
    \;(\text{deterministic signal})}
    \;+\;
    \underbrace{\eta\sum_{t=0}^{T-1}\gamma^{T-1-t}\,\psi_t\,\xi_t}
    _{\textstyle N_T \;(\text{noise})}.
\end{equation}
The noise term satisfies, for every $\varepsilon>0$,
\begin{equation}
    \Prob\bigl(|N_T| \geq \varepsilon\bigr)
    \;\leq\;
    2\exp\!\left(-\,\frac{\varepsilon^2(1-\gamma^2)}{8\eta^2}\right),
    \label{eq:noise-concentration}
\end{equation}
so that, with probability at least $1-\delta$,
\begin{equation}
    |N_T| \;\leq\; \eta\sqrt{\frac{8\log(2/\delta)}{1-\gamma^2}}
    \qquad \text{for every } T,
    \label{eq:noise-bound}
\end{equation}
uniformly in the horizon $T$.
\end{lemma}

\begin{proof}
Unrolling the recursion of Eq.~\ref{eq:sigma_mem_update} from
$\mathbf{M}_p^{(0)}=\mathbf{0}$ gives
$\mathbf{M}_p^{(T)} = \eta\sum_{t=0}^{T-1}\gamma^{T-1-t}c_{p,t}\,
\boldsymbol{\phi}(\mathbf{x}_t)\boldsymbol{\phi}(\mathbf{x}_t)^\top$.
Projecting onto $\boldsymbol{\phi}^\star$ and substituting
$c_{p,t}=\mu_p(\mathbf{x}_t)+\xi_t$ gives the stated decomposition, using
$\bigl(\boldsymbol{\phi}(\mathbf{x}_t)^\top\boldsymbol{\phi}^\star
\bigr)^2=\psi_t$.

For the concentration bound, write $N_T = \sum_{t=0}^{T-1} Y_t$ with
$Y_t = \eta\gamma^{T-1-t}\psi_t\,\xi_t$. The $Y_t$ are independent
(since the $\xi_t$ are), each has mean zero, and
$Y_t \in \bigl[-2\eta\gamma^{T-1-t}\psi_t,\,
2\eta\gamma^{T-1-t}\psi_t\bigr]$, a bounded range of length
$4\eta\gamma^{T-1-t}\psi_t$. Hoeffding's inequality for sums of
independent bounded random variables gives
\begin{equation}
    \Prob(|N_T|\geq\varepsilon)
    \;\leq\;
    2\exp\!\left(
      -\frac{2\varepsilon^2}
      {\sum_{t=0}^{T-1}\bigl(4\eta\gamma^{T-1-t}\psi_t\bigr)^2}
    \right)
    =
    2\exp\!\left(
      -\frac{\varepsilon^2}
      {8\eta^2\sum_{t=0}^{T-1}\gamma^{2(T-1-t)}\psi_t^2}
    \right).
\end{equation}
Since $\psi_t\in[0,1]$,
$\sum_{t=0}^{T-1}\gamma^{2(T-1-t)}\psi_t^2
\leq \sum_{k=0}^{T-1}\gamma^{2k} \leq \frac{1}{1-\gamma^2}$,
which gives Eq.~\ref{eq:noise-concentration}. Inverting the tail bound at
level $\delta$ gives Eq.~\ref{eq:noise-bound}.
\end{proof}
Lemma~\ref{lem:signal-noise} makes the qualitative claim in Section~\ref{sec:method:sigma_state} precise: however long the memory
runs, the noise contribution to any fixed direction is bounded by $O\bigl(\eta/\sqrt{1-\gamma^2}\bigr)$ with high probability, which does
not grow with $T$. Whether $\mathbf{M}_p$ becomes a reliable competence signal along $\boldsymbol{\phi}^\star$ therefore depends entirely on
whether the deterministic signal $S_T$ grows to exceed this fixed noise floor, which is exactly what persistent alignment guarantees, as we show next.

\textbf{Step 3: Persistent Alignment and Proof of Theorem~\ref{thm:sigma-mem-snr-main}}
\begin{corollary}[Stationary signal-to-noise ratio under persistent alignment]
\label{cor:snr}
Suppose the task direction is persistently aligned with $\boldsymbol{\phi}^\star$ ($\psi_t=1$ for all $t$) and peer $p$ has a
constant competence bias $\mu_p(\mathbf{x}_t)=\mu_p$ in this direction. Then the deterministic signal satisfies
\begin{equation}
    S_T = \eta\mu_p\,\frac{1-\gamma^T}{1-\gamma}
    \;\xrightarrow[T\to\infty]{}\;
    S_\infty = \frac{\eta\mu_p}{1-\gamma},
\end{equation}
while the noise bound of Eq.~\ref{eq:noise-bound} stabilizes at
$O\bigl(\eta/\sqrt{1-\gamma^2}\bigr)$. The resulting stationary
signal-to-noise ratio is
\begin{equation}
    \mathrm{SNR}(\gamma)
    \;=\;
    \frac{S_\infty}{\eta/\sqrt{1-\gamma^2}}
    \;=\;
    \mu_p\sqrt{\frac{1+\gamma}{1-\gamma}},
\end{equation}
which is strictly increasing in $\gamma$ and diverges as $\gamma\to 1^-$.
\end{corollary}

\begin{proof}
Immediate from Lemma~\ref{lem:signal-noise} with $\psi_t\equiv 1$ and
$\mu_p(\mathbf{x}_t)\equiv\mu_p$: $S_T$ is a geometric series summing to
$\eta\mu_p(1-\gamma^T)/(1-\gamma)$, which tends to $\eta\mu_p/(1-\gamma)$
as $T\to\infty$. The stationary noise scale follows by taking
$T\to\infty$ in Eq.~\ref{eq:noise-bound}, where the bound (up to the
$\sqrt{\log(2/\delta)}$ factor) is $\eta/\sqrt{1-\gamma^2}$. Dividing
gives $\mathrm{SNR}(\gamma) = \frac{\eta\mu_p/(1-\gamma)}
{\eta/\sqrt{1-\gamma^2}} = \mu_p\sqrt{(1-\gamma^2)}/(1-\gamma)
= \mu_p\sqrt{(1+\gamma)/(1-\gamma)}$, which is increasing in $\gamma$ on
$[0,1)$ since $(1+\gamma)/(1-\gamma)$ is increasing, and tends to
$+\infty$ as $\gamma\to 1^-$.
\end{proof}
\begin{remark}[Memory-length/adaptivity trade-off]
Corollary~\ref{cor:snr} shows that increasing $\gamma$ strictly improves the stationary signal-to-noise ratio under a constant competence bias $\mu_p$.
This trade-off is one-sided in our model because
Assumption~\ref{assumption:competence-model} treats $\mu_p$ as
time-invariant; if a peer's competence drifts over time, a larger
$\gamma$ also slows how quickly $\mathbf{M}_p$ forgets outdated evidence and adapts to the new competence level.
\end{remark}

It remains to connect the scalar projection $s_T$ back to the spectrum of $\mathbf{M}_p^{(T)}$ itself, via the variational characterization of eigenvalues, which completes the proof of Theorem~\ref{thm:sigma-mem-snr-main}.

\begin{proof}[Proof of Theorem~\ref{thm:sigma-mem-snr-main}]
Since $\mathbf{M}_p^{(T)}$ is symmetric, the Rayleigh--Ritz theorem gives
\begin{equation}
    \lambda_{\max}\bigl(\mathbf{M}_p^{(T)}\bigr)
    \;\geq\;
    \boldsymbol{\phi}^{\star\top}\mathbf{M}_p^{(T)}\boldsymbol{\phi}^\star
    \;=\; s_T
\end{equation}
for the fixed unit vector $\boldsymbol{\phi}^\star$. By
Lemma~\ref{lem:signal-noise} and Eq.~\ref{eq:noise-bound},
$s_T = S_T + N_T \geq S_T - \eta\sqrt{8\log(2/\delta)/(1-\gamma^2)}$ with
probability at least $1-\delta$. Substituting $S_T$ from
Corollary~\ref{cor:snr} gives
\begin{equation}
    \lambda_{\max}\bigl(\mathbf{M}_p^{(T)}\bigr)
    \;\geq\;
    \eta\mu_p\,\frac{1-\gamma^T}{1-\gamma}
    \;-\;
    \eta\sqrt{\frac{8\log(2/\delta)}{1-\gamma^2}}
\end{equation}
with probability at least $1-\delta$, which is the bound stated in
Theorem~\ref{thm:sigma-mem-snr-main}. The noise term is independent of
$T$ by construction, and the signal term saturates at
$\eta\mu_p/(1-\gamma)$ with stationary signal-to-noise ratio
$\mu_p\sqrt{(1+\gamma)/(1-\gamma)}$, increasing in $\gamma$, by
Corollary~\ref{cor:snr}. This proves all parts of the theorem.
\end{proof}
Thus, once $T$ is large enough that $\eta\mu_p(1-\gamma^T)/(1-\gamma)$ exceeds the noise floor $\eta\sqrt{8\log(2/\delta)/(1-\gamma^2)}$, equivalently, once $\mathrm{SNR}(\gamma)\sqrt{1-\gamma^{2T}}$ exceeds $\sqrt{8\log(2/\delta)}$, the leading eigenvalue of $\mathbf{M}_p$ is, with high probability, bounded away from zero by an amount reflecting peer $p$'s true competence $\mu_p$ in that direction, and not merely an artifact of event-level noise.

\section{Dataset Construction}
\label{appendix:dataset}
\subsection{Training Dataset}
\label{appendix:training_dataset}

We train \textbf{$\Sigma$-Mem} on a fixed offline dataset containing 2,963 task events from three competence domains: mathematics, retrieval-augmented question answering, and coding. Specifically, the dataset contains 1,000 examples from the
GSM8K~\citep{cobbe2021training} training split, 1,000 examples from the SQuAD~\citep{rajpurkar2016squad} training split, and 963 examples fromthe APPS~\citep{hendrycks2021measuring} training split. For each task, we independently generate one response from each of three fixed peer models: Gemma-3-4B-it~\citep{gemmateam2025gemma3technicalreport}, Phi-4-mini-instruct~\citep{abouelenin2025phi}, and Qwen2.5-Coder-7B-Instruct~\citep{yang2025qwen3}. Each peer response is then evaluated to obtain a binary correctness label used for event-level memory updates. Table~\ref{tab:training_peer_accuracy} summarizes the correctness of each peer on the training dataset.

\begin{table*}[!h]
\centering
\caption{Peer correctness on the training dataset. Each entry reports accuracy. The best performing peer on each dataset is highlighted  in \textbf{bold}.}
\label{tab:training_peer_accuracy}
\setlength{\tabcolsep}{7pt}
\renewcommand{\arraystretch}{1.15}

\begin{tabular}{l|c|ccc}
\Xhline{1.2pt}
\rowcolor{orange!15}
\textbf{Dataset}
& \textbf{N}
& Gemma-3-4B-it
& Phi-4-mini
& Qwen2.5-Coder \\
\Xhline{1pt}

GSM8K~\citep{cobbe2021training}
& 1,000
& \textbf{81.00\%}
& 61.90\%
& 71.00\% \\

SQuAD~\citep{rajpurkar2016squad}
& 1,000
& 12.80\%
& \textbf{83.70\%}
& 50.90\% \\

APPS~\citep{hendrycks2021measuring}
& 963
& \textbf{25.86\%}
& 8.10\%
& 17.65\% \\

\hline

Overall
& 2,963
& 40.06\%
& \textbf{51.77\%}
& 46.88\% \\

\Xhline{1.2pt}
\end{tabular}
\end{table*}

\subsection{Counterfactual Attack Dataset}
\label{appendix:counterfactual_dataset}

We construct a counterfactual attack dataset from nine benchmarks
covering mathematics, retrieval-augmented question answering (RAG), and coding.
The mathematics subset is drawn from
MATH-500~\citep{hendrycks2021measuring,lightman2024let},
AMC23~\citep{yang2024qwen2},
OlympiadBench~\citep{he2024olympiadbench}, and
CollegeMath~\citep{tang2024mathscale};
the RAG subset is drawn from
HotpotQA~\citep{yang2018hotpotqa} and
TriviaQA~\citep{joshi2017triviaqa};
and the coding subset is drawn from
HumanEval~\citep{chen2021evaluating},
MBPP~\citep{austin2021program}, and
LiveCodeBench~\citep{jain2025livecodebench}.
Each counterfactual split contains the same 2,685 tasks, including
1,216 mathematics tasks, 1,075 RAG tasks, and 394 coding tasks. 

For each constituent dataset, we first measure the clean accuracy of the three
peers and designate the best-performing model as the strong peer. An event is
eligible for counterfactual construction if the strong peer answers correctly
and at least one of the remaining peers answers incorrectly. We then swap the
complete response and correctness label of the strong peer with those of one
incorrect weak peer, while keeping the peer identities and metadata unchanged.
The attack ratio $r$ denotes the proportion of eligible events that are swapped
within each dataset. We evaluate
$r\in\{0\%,50\%,70\%,90\%\}$, where CF@0 is the clean baseline. All four splits
share the same underlying task events. Finally, we pool the nine datasets and
apply a deterministic shuffle to form a single mixed-domain evaluation stream.

\begin{table*}[!h]
\centering
\caption{Peer accuracy under different counterfactual attack ratios. All splits
contain the same 2,685 task events. The best performing peer in each split is
highlighted in \textbf{bold}.}
\label{tab:counterfactual_peer_accuracy}
\setlength{\tabcolsep}{7pt}
\renewcommand{\arraystretch}{1.15}

\begin{tabular}{l|c|ccc}
\Xhline{1.2pt}
\rowcolor{orange!15}
\textbf{CF Ratio}
& \textbf{N}
& Gemma-3-4B-it
& Phi-4-mini
& Qwen2.5-Coder \\
\Xhline{1pt}

CF@0
& 2,685
& \textbf{55.68\%}
& 46.93\%
& 48.42\% \\

CF@50
& 2,685
& 44.69\%
& 43.43\%
& \textbf{62.91\%} \\

CF@70
& 2,685
& 40.30\%
& 41.08\%
& \textbf{69.65\%} \\

CF@90
& 2,685
& 35.87\%
& 39.37\%
& \textbf{75.79\%} \\

\Xhline{1.2pt}
\end{tabular}
\end{table*}

To evaluate generalization to larger peer sets, we extend the original
three-peer evaluation stream with two unseen peer models,
Llama-3.2-3B-Instruct~\citep{grattafiori2024llama} and BitCPM-CANN-3B~\citep{team2025minicpm4}. We construct a four-peer setting by
adding Llama-3.2-3B-Instruct and a five-peer setting by further adding
BitCPM-CANN-3B, without modifying the learned parameters of $\Sigma$-Mem. Both
settings use the same 2,685 task events, domain composition, and counterfactual
attack ratios as the original three-peer evaluation. The peer accuracies for the
two enlarged peer pools are reported in
Tables~\ref{tab:four_peer_accuracy} and~\ref{tab:five_peer_accuracy}.

\begin{table*}[t]
\centering
\caption{Peer accuracy under different counterfactual attack ratios in the
four-peer setting. All splits contain the same 2,685 task events. The
best-performing peer in each split is highlighted in \textbf{bold}.}
\label{tab:four_peer_accuracy}
\setlength{\tabcolsep}{7pt}
\renewcommand{\arraystretch}{1.15}

\begin{tabular}{l|c|cccc}
\Xhline{1.2pt}
\rowcolor{orange!15}
\textbf{CF Ratio}
& \textbf{N}
& Gemma-3-4B-it
& Phi-4-mini
& Qwen2.5-Coder
& Llama-3.2t \\
\Xhline{1pt}

CF@0
& 2,685
& \textbf{55.68\%}
& 46.93\%
& 48.42\%
& 38.70\% \\

CF@50
& 2,685
& 40.45\%
& 37.02\%
& \textbf{57.06\%}
& 55.20\% \\

CF@70
& 2,685
& 34.30\%
& 33.00\%
& 60.56\%
& \textbf{61.86\%} \\

CF@90
& 2,685
& 28.27\%
& 29.01\%
& 63.95\%
& \textbf{68.49\%} \\

\Xhline{1.2pt}
\end{tabular}
\end{table*}

\begin{table*}[!h]
\centering
\caption{Peer accuracy under different counterfactual attack ratios in the
five-peer setting. All splits contain the same 2,685 task events. The
best-performing peer in each split is highlighted in \textbf{bold}.}
\label{tab:five_peer_accuracy}
\setlength{\tabcolsep}{5pt}
\renewcommand{\arraystretch}{1.15}

\begin{tabular}{l|c|ccccc}
\Xhline{1.2pt}
\rowcolor{orange!15}
\textbf{CF Ratio}
& \textbf{N}
& Gemma-3-4B-it
& Phi-4-mini
& Qwen2.5-Coder
& Llama-3.2
& BitCPM-CANN-3B \\
\Xhline{1pt}

CF@0
& 2,685
& \textbf{55.68\%}
& 46.93\%
& 48.42\%
& 38.70\%
& 46.78\% \\

CF@50
& 2,685
& 38.29\%
& 34.56\%
& 53.45\%
& 51.17\%
& \textbf{59.03\%} \\

CF@70
& 2,685
& 31.36\%
& 29.53\%
& 55.53\%
& 56.09\%
& \textbf{63.99\%} \\

CF@90
& 2,685
& 24.36\%
& 24.47\%
& 57.54\%
& 61.04\%
& \textbf{69.09\%} \\

\Xhline{1.2pt}
\end{tabular}
\end{table*}

\subsection{Generalization Beyond Training Domains}
\label{appendix:ood_datasets}

Although $\Sigma$-Mem is trained only on mathematics, RAG, and coding events, it can
map an unseen task into the learned competence-direction space through
$\boldsymbol{\phi}(\mathbf{x})$. This allows historical reliability evidence from
related training domains to be reused for new tasks. To evaluate this capability,
we test the three fixed peers on out-of-distribution (OOD) benchmarks covering
physical commonsense reasoning (PIQA~\citep{bisk2020piqa}), broad-domain knowledge understanding
(MMLU~\citep{hendrycks2020measuring}), science question answering (OpenBookQA~\citep{mihaylov2018can} and SciQ~\citep{welbl2017crowdsourcing}), complex reasoning
(BBH~\citep{suzgun2023challenging}), and natural-language understanding (SuperGLUE~\citep{wang2019superglue}).
Table~\ref{tab:ood_peer_accuracy} summarizes peer accuracy across these unseen
domains.

\begin{table*}[!h]
\centering
\caption{Peer accuracy on benchmarks beyond the training domains. The SuperGLUE
result is computed by pooling its six evaluated subsets by instance count. The
best-performing peer on each benchmark is highlighted in bold.}
\label{tab:ood_peer_accuracy}
\setlength{\tabcolsep}{7pt}
\renewcommand{\arraystretch}{1.15}

\begin{tabular}{l|c|ccc}
\Xhline{1.2pt}
\rowcolor{orange!15}
\textbf{Dataset}
& \textbf{N}
& Gemma-3-4B-it
& Phi-4-mini
& Qwen2.5-Coder \\
\Xhline{1pt}

PIQA~\citep{bisk2020piqa}
& 1,838
& 76.99\%
& 80.14\%
& \textbf{84.44\%} \\

MMLU~\citep{hendrycks2020measuring}
& 1,531
& 54.08\%
& \textbf{64.08\%}
& 56.17\% \\

OpenBookQA~\citep{mihaylov2018can}
& 500
& 72.80\%
& 76.40\%
& \textbf{79.80\%} \\

SciQ~\citep{welbl2017crowdsourcing}
& 1,000
& 87.10\%
& 90.10\%
& \textbf{92.10\%} \\

BBH~\citep{suzgun2023challenging}
& 6,511
& \textbf{28.67\%}
& 10.69\%
& 12.78\% \\

SuperGLUE~\citep{wang2019superglue}
& 6,023
& 77.47\%
& 79.15\%
& \textbf{80.67\%} \\

\Xhline{1.2pt}
\end{tabular}
\end{table*}

\section{Direct Memory Readout under Counterfactual Streams}
\label{appendix:role}
To further illustrate how recorded historical competence evidence works for peer selection, we conduct a separation experiment on the CF streams. For each counterfactual ratio, we directly read out the peer-specific memory matrices and perform \textbf{M-route}: the current problem is encoded into a competence direction $\boldsymbol{\phi}(\mathbf{x}_t)$, each peer receives a reliability score $s_{p,t}=\boldsymbol{\phi}(\mathbf{x}_t)^\top \mathbf{M}^{(t)}_p\boldsymbol{\phi}(\mathbf{x}_t)$, and the peer with the largest score is selected. This route does not use the peer responses during selection; the correctness labels are only used afterward to update the memory.

We additionally compare against a conventional \textbf{Bayesian Beta B1} reputation baseline. B1 maintains one global posterior $\operatorname{Beta}(\alpha_p,\beta_p)$ for each peer. Starting from $\operatorname{Beta}(1,1)$, we first warm the posteriors on the same 2,963-event mixed training set used by \textbf{$\Sigma$-Mem}. For each CF ratio, we initialize B1 from this mixed-trained state, select the peer with the largest posterior mean $\alpha_p/(\alpha_p+\beta_p)$, and then update all peers after the current correctness feedback is revealed:
\begin{equation}
\alpha_p \leftarrow \gamma\alpha_p+r_p,
\qquad
\beta_p \leftarrow \gamma\beta_p+(1-r_p),
\end{equation}
where $r_p\in\{0,1\}$ and $\gamma=0.9$. B1 is response-blind and model-independent: it receives neither the current question nor any peer response, and uses only accumulated correctness feedback. The results are shown in Table~\ref{tab:CF_Mresults}.

\begin{table*}[!h]
\centering
\caption{Counter Factual (CF) Attack results under different CF@ratios. \texttt{Acc} denotes the task accuracy obtained using the selected peer response, and \texttt{p1(Gemma)/p2(Phi)/p3(Qwen-Coder)} reports selection counts for the three peers. \texttt{M-route} directly selects peers from the question-derived memory readout without observing peer responses. Bayesian Beta B1 instead uses a mixed-trained global correctness posterior and is model-independent, so it is reported once in the final row. We \textcolor{darkred}{\textbf{highlight}} the best result within each center-model block and CF@ratio.}
\label{tab:CF_Mresults}
\setlength{\tabcolsep}{3pt}
\renewcommand{\arraystretch}{1.15}
\resizebox{\textwidth}{!}{
\begin{tabular}{lcccccccc}
\toprule
\multirow{2}{*}{Model}
& \multicolumn{2}{c}{CF@0}
& \multicolumn{2}{c}{CF@50}
& \multicolumn{2}{c}{CF@70}
& \multicolumn{2}{c}{CF@90} \\
\cmidrule(lr){2-3}
\cmidrule(lr){4-5}
\cmidrule(lr){6-7}
\cmidrule(lr){8-9}
& Acc & p1/p2/p3
& Acc & p1/p2/p3
& Acc & p1/p2/p3
& Acc & p1/p2/p3 \\
\midrule

Qwen3-0.6B
& 61.01\% & 1837/192/656
& 52.89\% & 1828/222/635
& 50.50\% & 1838/228/619
& 46.22\% & 1830/229/626 \\

\enspace + $\Sigma$ w/o G
& 64.17\% & 1801/355/529
& 55.38\% & 617/511/1557
& 61.19\% & 261/714/1710
& 68.64\% & 82/968/1635 \\

\enspace + $\Sigma$ w/ G
& 63.87\% & 1817/347/521
& 56.42\% & 614/396/1675
& 62.72\% & 286/518/1881
& 71.10\% & 81/611/1993 \\

\enspace + M-route
& \textcolor{darkred}{\textbf{65.51\%}} & 1102/1063/520
& \textcolor{darkred}{\textbf{57.88\%}} & 471/428/1786
& \textcolor{darkred}{\textbf{66.67\%}} & 251/310/2124
& \textcolor{darkred}{\textbf{76.91\%}} & 84/342/2259 \\

\midrule

Qwen3-4B
& 61.71\% & 930/339/1416
& \textcolor{darkred}{\textbf{65.55\%}} & 930/429/1326
& 65.77\% & 944/460/1281
& 67.86\% & 955/482/1248 \\

\enspace + $\Sigma$ w/o G
& 63.35\% & 1434/589/662
& 61.15\% & 451/390/1844
& 66.41\% & 238/363/2084
& 72.74\% & 129/351/2205 \\

\enspace + $\Sigma$ w/ G
& 64.51\% & 1290/542/853
& 63.61\% & 447/344/1894
& \textcolor{darkred}{\textbf{67.64\%}} & 265/300/2120
& 73.30\% & 143/287/2255 \\

\enspace + M-route
& \textcolor{darkred}{\textbf{64.58\%}} & 1122/1036/527
& 56.69\% & 470/411/1804
& 65.40\% & 241/301/2143
& \textcolor{darkred}{\textbf{75.42\%}} & 82/335/2268 \\

\midrule

Qwen3-8B
& 66.89\% & 1068/572/1045
& \textcolor{darkred}{\textbf{67.67\%}} & 1006/573/1106
& 66.26\% & 1008/566/1111
& 66.07\% & 985/545/1155 \\

\enspace + $\Sigma$ w/o G
& 67.75\% & 1392/624/669
& 62.53\% & 535/547/1603
& 64.80\% & 307/552/1826
& 68.34\% & 227/571/1887 \\

\enspace + $\Sigma$ w/ G
& 67.67\% & 1384/558/743
& 63.76\% & 500/462/1723
& 66.33\% & 230/493/1962
& 69.16\% & 129/494/2062 \\

\enspace + M-route
& \textcolor{darkred}{\textbf{68.75\%}} & 1123/1050/512
& 57.91\% & 451/409/1825
& \textcolor{darkred}{\textbf{66.52\%}} & 258/269/2158
& \textcolor{darkred}{\textbf{75.72\%}} & 100/261/2324 \\

\midrule

Qwen3.5-4B
& 67.15\% & 1360/438/887
& \textcolor{darkred}{\textbf{66.15\%}} & 1271/526/888
& 64.69\% & 1258/563/864
& 63.87\% & 1211/565/909 \\

\enspace + $\Sigma$ w/o G
& 70.99\% & 1745/583/357
& 60.89\% & 463/385/1837
& 66.89\% & 179/457/2049
& 73.15\% & 70/696/1919 \\

\enspace + $\Sigma$ w/ G
& \textcolor{darkred}{\textbf{72.77\%}} & 1619/596/470
& 63.17\% & 451/331/1903
& \textcolor{darkred}{\textbf{68.60\%}} & 182/314/2189
& 75.46\% & 71/288/2326 \\

\enspace + M-route
& 71.99\% & 1236/956/493
& 57.39\% & 435/390/1860
& 67.49\% & 225/246/2214
& \textcolor{darkred}{\textbf{76.72\%}} & 84/218/2383 \\

\midrule

Qwen3.5-9B
& 66.48\% & 1041/582/1062
& \textcolor{darkred}{\textbf{68.94\%}} & 1144/528/1013
& 69.65\% & 1196/495/994
& 70.99\% & 1193/481/1011 \\

\enspace + $\Sigma$ w/o G
& 71.81\% & 1771/514/400
& 62.83\% & 469/359/1857
& 69.24\% & 210/343/2132
& 75.46\% & 137/329/2219 \\

\enspace + $\Sigma$ w/ G
& 71.84\% & 1629/498/558
& 65.96\% & 459/316/1910
& \textcolor{darkred}{\textbf{71.40\%}} & 194/284/2207
& \textcolor{darkred}{\textbf{76.95\%}} & 95/214/2376 \\

\enspace + M-route
& \textcolor{darkred}{\textbf{77.36\%}} & 1266/1005/414
& 56.20\% & 446/412/1827
& 67.08\% & 280/291/2114
& 74.79\% & 163/276/2246 \\

\midrule

Bayesian Beta B1
& 51.43\% & 1397/667/621
& 58.18\% & 328/282/2075
& 67.26\% & 96/130/2459
& 75.23\% & 31/28/2626 \\

\bottomrule
\end{tabular}
}
\end{table*}

The results reflect two factors: the judgment ability of the central model and the validity of the historical signal in the stream.
For the Qwen3 center models, whose response-level judgment is weaker, M-route matches or exceeds $\Sigma$ w/G in 9 of 12 settings, with margins up to 5.81 points (Qwen3-0.6B, CF@90) and 6.56 points (Qwen3-8B, CF@90). $\Sigma$-Mem must pass its reliability evidence through the central model's judgment, while M-route does not, so a weak judge dilutes the benefit of the memory. For the stronger Qwen3.5 backbones, the comparison flips: $\Sigma$ w/G surpasses M-route in 6 of 8 settings, and where M-route remains ahead, the margin shrinks. For example, at CF@90 the M-route advantage is 1.26 points for Qwen3.5-4B, compared with 2.12 and 6.56 points for Qwen3-4B and Qwen3-8B at the same split. A stronger central model therefore converts the same memory state into better decisions by adding reliable response-level judgment on top of it.
The main exception across all central models is CF@50, where M-route falls to 56.20\%--57.91\% and lies below $\Sigma$ w/G for every central model larger than 0.6B. \textbf{This is not a failure of the memory mechanism but its faithful behavior: the CF@50 stream is ambiguous by construction, since half of the history supports the original peer preference while the other half supports the counterfactual shift, and a faithful memory cannot produce confident guidance from evidence that contains none.} \textbf{Therefore, using historical reliability evidence to guide peer selection can become misleading when the stream itself provides ambiguous reliability evidence.} This reinforces our interpretation that \textbf{$\Sigma$-Mem integrates the historical evidence it receives, and motivates decision rules that balance memory-based reliability with the center model's current response-level judgment}. Indeed, at CF@50 that response-level judgment is the only informative signal, which explains why $\Sigma$ w/G retains a clear advantage over M-route for every central model above 0.6B.

The Bayesian comparison provides a complementary control. Beta B1 obtains 51.43\%, 58.18\%, 67.26\%, and 75.23\% from CF@0 to CF@90. At CF@0, $\Sigma$ w/G exceeds B1 by 12.44--21.34 points across all center models, showing the limitation of representing heterogeneous mathematics, RAG, and coding reliability with only one global posterior per peer. At CF@50, $\Sigma$ w/G also exceeds B1 for four of the five center models. As the counterfactual shift becomes dominant, B1 adapts by selecting Qwen-Coder 2,459 times at CF@70 and 2,626 times at CF@90, making it competitive with both M-route and $\Sigma$-Mem. Nevertheless, the Qwen3.5 variants with $\Sigma$ w/G remain above B1 at both CF@70 and CF@90. These results show that a global Bayesian success-rate estimator is strong when one peer becomes globally dominant, while task-conditioned memory and current-response judgment are more effective when reliability varies across tasks or the global history is ambiguous.

Overall, memory-only routing is preferable when the judge is weak or the reliability pattern in the stream is stable, while response-based judgment becomes essential when history is ambiguous, and its value grows with the strength of the central model. The Beta B1 comparison further shows that the gains of $\Sigma$-Mem cannot be reduced to conventional global Bayesian reputation tracking.

% $$
%   S(a)=\sum_{p:\operatorname{canon}(a_p)=a}\rho_p,
%   \qquad
%   \hat a=\arg\max_a S(a)
% $$

\section{SuperGLUE breakdown}
The aggregate SuperGLUE result in Table~\ref{tab:ood_peer_accuracy} pools the six
evaluated validation subsets according to their number of instances.
Table~\ref{tab:superglue_peer_breakdown} reports peer accuracy on each subset
separately.

\begin{table*}[!h]
\centering
\caption{Peer accuracy on the six evaluated SuperGLUE subsets. The
best-performing peer on each subset is highlighted in bold.}
\label{tab:superglue_peer_breakdown}
\setlength{\tabcolsep}{7pt}
\renewcommand{\arraystretch}{1.15}

\begin{tabular}{l|c|ccc}
\Xhline{1.2pt}
\rowcolor{orange!15}
\textbf{SuperGLUE Subset}
& \textbf{N}
& Gemma-3-4B-it
& Phi-4-mini
& Qwen2.5-Coder \\
\Xhline{1pt}

WiC
& 638
& 54.70\%
& \textbf{62.70\%}
& 59.56\% \\

RTE
& 277
& 54.15\%
& 30.32\%
& \textbf{88.81\%} \\

CB
& 56
& 48.21\%
& 5.36\%
& \textbf{82.14\%} \\

COPA
& 100
& 93.00\%
& 93.00\%
& \textbf{97.00\%} \\

WSC
& 104
& 37.50\%
& 43.27\%
& \textbf{47.12\%} \\

MultiRC
& 4,848
& 82.67\%
& \textbf{85.44\%}
& 83.35\% \\

\Xhline{1.2pt}
\end{tabular}
\end{table*}

\begin{figure}[!h]
\vspace{-1.0em}
\centering
\includegraphics[width=0.9\linewidth]{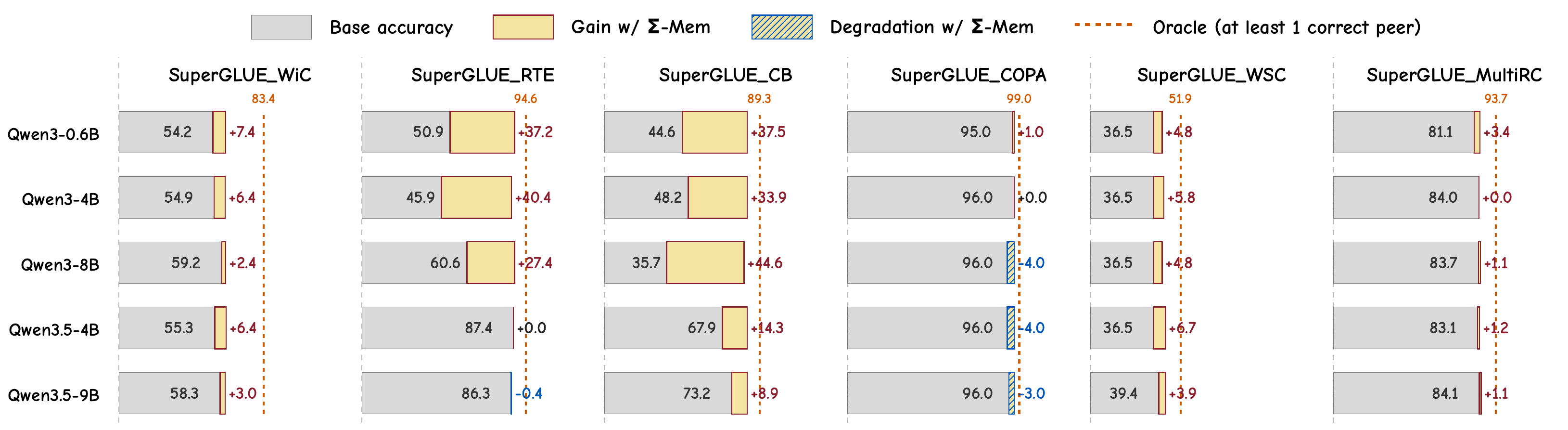}
\caption{SuperGLUE dataset breakdown across six subtasks. Each horizontal bar decomposes performance into the \textbf{\textcolor{gray}{base model accuracy}} and the change introduced by \textbf{$\Sigma$-Mem}: positive gains are shown in \textcolor{darkred}{red}, while degradations are shown in 
\textcolor{blue}{blue}. The \textcolor{orange}{orange dashed line} denotes the oracle accuracy obtained when at least one peer provides a correct answer, which is the upper bound of accuracy.}
\label{fig:superglue}
\end{figure}

As shown in Figure~\ref{fig:superglue}, \textbf{$\Sigma$-Mem improves aggregation accuracy in most cases}. The gains are particularly large for the Qwen3 models on SuperGLUE-RTE and SuperGLUE-CB, where the base center models fail to identify the correct peer, while historical competence and peer-relationship evidence substantially improve selection. A small degradation appears in a few cases where the base model is already highly accurate, as historical steering can slightly alter an otherwise correct judgment. Nevertheless, after applying \textbf{$\Sigma$-Mem}, performance across nearly all cases moves closer to the oracle upper bound, demonstrating its overall effectiveness for more reliable peer aggregation.

\section{Validation of Historical Competence in $\mathbf{M}$}
\label{appendix:validation}
This section reports the per-benchmark results for the direct-readout experiments of Sec.~\ref{sec:coordination-role}. The decision rules are defined in Sec.~\ref{sec:method:direct_readouts}.

We evaluate this diagnostic on the same OOD datasets and three fixed peers described in Sec.~\ref{appendix:ood_datasets} using
Qwen3-0.6B, Qwen3-4B, Qwen3-8B, Qwen3.5-4B, and Qwen3.5-9B as central models. These comparisons isolate the contribution of task-conditioned historical competence evidence itself provides reliable signals. The experimental results are as follows.

\begin{table*}[!h]
\centering
\caption{\textbf{\textcolor{Blue}{Qwen 3- 0.6B}}: Diagnostic results of directly reading historical
competence evidence from $\mathbf{M}_p$. The best-performing result on each benchmark is highlighted in \textbf{bold}.}
\label{tab:m_readout_qwen06b}
\setlength{\tabcolsep}{7pt}
\renewcommand{\arraystretch}{1.15}

\begin{tabular}{l|c|ccccc}
\Xhline{1.2pt}
\rowcolor{orange!15}
\textbf{OOD Group}
& \textbf{N}
& \textbf{Base}
& \textbf{$\Sigma$ w/ G}
& \textbf{Majority}
& \textbf{M Route}
& \textbf{M Vote} \\
\Xhline{1pt}

PIQA
& 1,838
& 80.69\%
& 81.12\%
& \textbf{83.73\%}
& 82.64\%
& \textbf{83.73\%} \\

MMLU
& 1,531
& 57.61\%
& 61.20\%
& \textbf{62.12\%}
& 60.42\%
& 61.01\% \\

OpenBookQA
& 500
& 74.00\%
& 79.00\%
& 81.00\%
& 78.60\%
& \textbf{81.20\%} \\

SciQ
& 1,000
& 88.10\%
& 88.50\%
& 92.50\%
& 90.60\%
& \textbf{92.60\%} \\

BBH
& 6,511
& 25.25\%
& \textbf{28.77\%}
& 22.78\%
& 28.54\%
& 27.54\% \\

SuperGLUE
& 6,023
& 75.99\%
& 80.38\%
& 82.78\%
& 82.30\%
& \textbf{83.10\%} \\

\Xhline{1pt}

\textbf{Overall}
& \textbf{17,403}
& 56.52\%
& 59.89\%
& 59.12\%
& 60.67\%
& \textbf{60.93\%} \\

\Xhline{1.2pt}
\end{tabular}
\end{table*}

\begin{table*}[!h]
\centering
\caption{\textbf{\textcolor{Blue}{Qwen3-4B}}:
Diagnostic results of directly reading historical competence evidence from
$\mathbf{M}_p$. The best-performing result on each benchmark is highlighted
in \textbf{bold}.}
\label{tab:m_readout_qwen3_4b}
\setlength{\tabcolsep}{7pt}
\renewcommand{\arraystretch}{1.15}

\begin{tabular}{l|c|ccccc}
\Xhline{1.2pt}
\rowcolor{orange!15}
\textbf{OOD Group}
& \textbf{N}
& \textbf{Base}
& \textbf{$\Sigma$ w/ G}
& \textbf{Majority}
& \textbf{M Route}
& \textbf{M Vote} \\
\Xhline{1pt}

PIQA
& 1,838
& 81.77\%
& 82.92\%
& \textbf{83.73\%}
& 82.64\%
& \textbf{83.73\%} \\

MMLU
& 1,531
& 60.29\%
& \textbf{62.51\%}
& 62.12\%
& 60.55\%
& 61.20\% \\

OpenBookQA
& 500
& 76.20\%
& 78.40\%
& 81.00\%
& 78.40\%
& \textbf{81.20\%} \\

SciQ
& 1,000
& 89.90\%
& 91.00\%
& 92.50\%
& 90.60\%
& \textbf{92.60\%} \\

BBH
& 6,511
& 20.38\%
& \textbf{28.66\%}
& 22.78\%
& 28.52\%
& 27.65\% \\

SuperGLUE
& 6,023
& 78.18\%
& 81.14\%
& 82.78\%
& 82.35\%
& \textbf{83.03\%} \\

\Xhline{1pt}

\textbf{Overall}
& \textbf{17,403}
& 55.98\%
& 60.54\%
& 59.12\%
& 60.68\%
& \textbf{60.96\%} \\

\Xhline{1.2pt}
\end{tabular}
\end{table*}

\begin{table*}[!h]
\centering
\caption{\textbf{\textcolor{Blue}{Qwen3-8B}}:
Diagnostic results of directly reading historical competence evidence from
$\mathbf{M}_p$. The best-performing result on each benchmark is highlighted
in \textbf{bold}.}
\label{tab:m_readout_qwen3_8b}
\setlength{\tabcolsep}{7pt}
\renewcommand{\arraystretch}{1.15}

\begin{tabular}{l|c|ccccc}
\Xhline{1.2pt}
\rowcolor{orange!15}
\textbf{OOD Group}
& \textbf{N}
& \textbf{Base}
& \textbf{$\Sigma$ w/ G}
& \textbf{Majority}
& \textbf{M Route}
& \textbf{M Vote} \\
\Xhline{1pt}

PIQA
& 1,838
& 82.92\%
& 83.57\%
& 83.73\%
& 82.59\%
& \textbf{83.84\%} \\

MMLU
& 1,531
& 60.42\%
& \textbf{62.90\%}
& 62.12\%
& 60.61\%
& 61.33\% \\

OpenBookQA
& 500
& 79.60\%
& 79.00\%
& 81.00\%
& 78.40\%
& \textbf{81.20\%} \\

SciQ
& 1,000
& 91.30\%
& 91.70\%
& \textbf{92.50\%}
& 90.50\%
& \textbf{92.50\%} \\

BBH
& 6,511
& 19.29\%
& 28.17\%
& 22.78\%
& \textbf{28.54\%}
& 27.66\% \\

SuperGLUE
& 6,023
& 79.00\%
& 81.07\%
& 82.78\%
& 82.32\%
& \textbf{83.05\%} \\

\Xhline{1pt}

\textbf{Overall}
& \textbf{17,403}
& 56.16\%
& 60.50\%
& 59.12\%
& 60.67\%
& \textbf{60.99\%} \\

\Xhline{1.2pt}
\end{tabular}
\end{table*}

\begin{table*}[!h]
\centering
\caption{\textbf{\textcolor{Blue}{Qwen3.5-4B}}:
Diagnostic results of directly reading historical competence evidence from
$\mathbf{M}_p$. The best-performing result on each benchmark is highlighted
in \textbf{bold}.}
\label{tab:m_readout_qwen35_4b}
\setlength{\tabcolsep}{7pt}
\renewcommand{\arraystretch}{1.15}

\begin{tabular}{l|c|ccccc}
\Xhline{1.2pt}
\rowcolor{orange!15}
\textbf{OOD Group}
& \textbf{N}
& \textbf{Base}
& \textbf{$\Sigma$ w/ G}
& \textbf{Majority}
& \textbf{M Route}
& \textbf{M Vote} \\
\Xhline{1pt}

PIQA
& 1,838
& 83.19\%
& 83.46\%
& 83.73\%
& 82.70\%
& \textbf{83.79\%} \\

MMLU
& 1,531
& 61.20\%
& \textbf{62.90\%}
& 62.12\%
& 60.61\%
& 60.94\% \\

OpenBookQA
& 500
& 79.60\%
& \textbf{81.40\%}
& 81.00\%
& 78.20\%
& \textbf{81.40\%} \\

SciQ
& 1,000
& 92.00\%
& \textbf{92.80\%}
& 92.50\%
& 90.50\%
& 92.50\% \\

BBH
& 6,511
& 27.46\%
& \textbf{29.15\%}
& 22.78\%
& 28.44\%
& 27.32\% \\

SuperGLUE
& 6,023
& 79.59\%
& 80.76\%
& 82.78\%
& 82.38\%
& \textbf{83.13\%} \\

\Xhline{1pt}

\textbf{Overall}
& \textbf{17,403}
& 59.56\%
& \textbf{60.87\%}
& 59.12\%
& 60.67\%
& 60.86\% \\

\Xhline{1.2pt}
\end{tabular}
\end{table*}

\begin{table*}[!h]
\centering
\caption{\textbf{\textcolor{Blue}{Qwen3.5-9B}}:
Diagnostic results of directly reading historical competence evidence from
$\mathbf{M}_p$. The best-performing result on each benchmark is highlighted
in \textbf{bold}.}
\label{tab:m_readout_qwen35_9b}
\setlength{\tabcolsep}{7pt}
\renewcommand{\arraystretch}{1.15}

\begin{tabular}{l|c|ccccc}
\Xhline{1.2pt}
\rowcolor{orange!15}
\textbf{OOD Group}
& \textbf{N}
& \textbf{Base}
& \textbf{$\Sigma$ w/ G}
& \textbf{Majority}
& \textbf{M Route}
& \textbf{M Vote} \\
\Xhline{1pt}

PIQA
& 1,838
& \textbf{84.55\%}
& 83.46\%
& 83.73\%
& 82.75\%
& 83.79\% \\

MMLU
& 1,531
& 62.83\%
& \textbf{64.27\%}
& 62.12\%
& 60.94\%
& 61.14\% \\

OpenBookQA
& 500
& 79.80\%
& 80.20\%
& \textbf{81.00\%}
& 78.00\%
& \textbf{81.00\%} \\

SciQ
& 1,000
& 92.20\%
& 92.10\%
& \textbf{92.50\%}
& 90.40\%
& 92.40\% \\

BBH
& 6,511
& 26.80\%
& \textbf{29.33\%}
& 22.78\%
& 28.46\%
& 27.34\% \\

SuperGLUE
& 6,023
& 80.79\%
& 81.67\%
& 82.78\%
& 82.37\%
& \textbf{83.13\%} \\

\Xhline{1pt}

\textbf{Overall}
& \textbf{17,403}
& 60.04\%
& \textbf{61.31\%}
& 59.12\%
& 60.69\%
& 60.86\% \\

\Xhline{1.2pt}
\end{tabular}
\end{table*}

\end{document}